\documentclass[11pt]{article}
\usepackage{graphicx}
\usepackage{fullpage}
\usepackage{microtype}
\usepackage{amsmath,amsthm,amssymb}
\usepackage{color}
\usepackage{verbatim}
\usepackage{float}
\usepackage{comment}
\usepackage{algorithm}
\usepackage{algpseudocode}
\usepackage{chngcntr}
\usepackage{apptools}

\begin{document}

\title{Fractal dimension and lower bounds for geometric problems}

\author{
Anastasios Sidiropoulos\thanks{Dept.~of Computer Science, University of Illinois at Chicago, \texttt{sidiropo@uic.edu}.}
\and
Kritika Singhal\thanks{Dept.~of Mathematics, The Ohio State University, \texttt{singhal.53@osu.edu}.}
\and
Vijay Sridhar\thanks{Dept.~of Computer Science \& Engineering, The Ohio State University, \texttt{sridhar.38@osu.edu}.}
}

\def\lf{\left\lfloor}   
\def\rf{\right\rfloor}
\def\lc{\left\lceil}   
\def\rc{\right\rceil}

\newtheorem{theorem}{Theorem}[section]
\newtheorem{lemma}[theorem]{Lemma}
\newtheorem{proposition}[theorem]{Proposition}
\newtheorem{claim}[theorem]{Claim}
\newtheorem{corollary}[theorem]{Corollary}
\newtheorem{definition}[theorem]{Definition}
\newtheorem{observation}[theorem]{Observation}
\newtheorem{fact}[theorem]{Fact}
\newtheorem{property}{Property}
\newtheorem{Remark}{Remark}[section]
\newtheorem{notation}{Notation}[section]
\newtheorem{example}{Example}[section]
\newtheorem{conjecture}{Conjecture}
\newtheorem{question}[conjecture]{Question}

\newcommand{\eps}{\varepsilon}
\newcommand{\tw}{\mathsf{tw}}
\newcommand{\rep}{\mathsf{rep}}
\newcommand{\pw}{\mathsf{pw}}
\newcommand{\dimF}{\dim_{\mathsf{f}}}
\newcommand{\dimD}{\dim_{\mathsf{d}}}
\newcommand{\diam}{\mathsf{diam}}
\newcommand{\dist}{\mathsf{dist}}
\newcommand{\volume}{\mathsf{volume}}
\newcommand{\ifac}{\mathsf{interface}}
\newcommand{\dbox}{\mathsf{box}}
\newcommand{\ball}{\mathsf{ball}}
\newcommand{\sphere}{\mathsf{sphere}}
\newcommand{\origin}{\mathbf{o}}
\newcommand{\OPT}{\mathsf{OPT}}
\newcommand{\XXX}{\textcolor{red}{XXX}}
\newcommand{\siep}{Sierpi\'{n}ski}
\newcommand{\ecp}{Exact Cover Problem}
\newcommand{\csp}{Constraint Satisfaction Problem}
\newcommand{\etal}{\textit{et al}.}

\AtAppendix{\counterwithin{theorem}{section}}


\clearpage

\date{}
\maketitle

\thispagestyle{empty}
\begin{abstract}
We study the complexity of geometric problems on spaces of low \emph{fractal dimension}.
It was recently shown by [Sidiropoulos \& Sridhar, SoCG 2017] that several problems admit improved solutions when the input is a pointset in Euclidean space with fractal dimension smaller than the ambient dimension.
In this paper we prove nearly-matching lower bounds, thus establishing nearly-optimal bounds for various problems as a function of the fractal dimension.

More specifically, we show that for any set of $n$ points in $d$-dimensional Euclidean space, of fractal dimension $\delta\in (1,d)$, for any $\eps>0$ and $c\geq 1$, any $c$-spanner must have treewidth at least $\Omega \left( \frac{n^{1-1/(\delta - \epsilon)}}{c^{d-1}} \right)$, matching the previous upper bound.
The construction used to prove this lower bound on the treewidth of spanners, can also be used to derive lower bounds on the running time of algorithms for various problems, assuming the Exponential Time Hypothesis.
We provide two prototypical results of this type:
\begin{itemize}
\item For any $\delta \in (1,d)$ and any $\eps >0$, $d$-dimensional Euclidean TSP on $n$ points with fractal dimension at most $\delta$ cannot be solved in time $2^{O\left(n^{1-1/(\delta - \eps)} \right)}$.
The best-known upper bound is $2^{O(n^{1-1/\delta} \log n)}$.

\item For any $\delta \in (1,d)$ and any $\eps >0$, the problem of finding $k$-pairwise non-intersecting $d$-dimensional unit balls/axis parallel unit cubes with centers having fractal dimension at most $\delta$ cannot be solved in time $f(k)n^{O \left(k^{1-1/(\delta - \eps)}\right)}$ for any computable function $f$.
The best-known upper bound is $n^{O(k^{1-1/\delta} \log n)}$.
\end{itemize}
The above results nearly match previously known upper bounds from [Sidiropoulos \& Sridhar, SoCG 2017], and generalize analogous lower bounds for the case of ambient dimension due to [Marx \& Sidiropoulos, SoCG 2014].
\end{abstract}

\pagebreak{}
\setcounter{page}{1}

\section{Introduction}

The \emph{curse of dimensionality} is a general phenomenon in computational geometry, asserting that the complexity of many problems increases rapidly with the dimension of the input.
Sets of fractional dimension can be used to model various processes and phenomena in science and engineering \cite{takayasu1990fractals}.
Recently, the complexity of various geometric optimization problems was studied as a function of the fractal dimension of the input \cite{sidiropoulos2017algorithmic}.
It was shown that, for several problems, improved algorithms can be obtained when the fractal dimension is smaller than the ambient dimension.

Interestingly, the algorithms obtained in \cite{sidiropoulos2017algorithmic} nearly match the best-possible algorithms for integral dimension.
In this paper, we give nearly-matching lower bounds, assuming the Exponential Time Hypotheis (ETH).
We remark that there are several different definitions of fractal dimension that can be considered.
Our results indicate that, for the case of Euclidean pointsets, the definition of fractal dimension we consider is the ``correct'' one for certain computational problems.
That is, it precisely generalizes the dependence of the running time on the ambient dimension.

\subsection{Our contribution}

We obtain nearly-optimal lower bounds for various 
prototypical geometric problems.
Our results are obtained via a general method that could be applicable to other problems.


\paragraph*{Spanners}
We begin with a lower bound on the treewidth of spanners.
It is known that any set of $n$ points in $\mathbb{R}^d$ admits a $(1+\eps)$-spanner of size $n (1/\eps)^{O(d)}$ \cite{salowe1991construction,vaidya1991sparse}.
This result has been generalized for the case of fractal dimension.
Specifically, it was shown in \cite{sidiropoulos2017algorithmic} that any $n$-point set in $O(1)$-dimensional Euclidean space,  of fractal dimension $\delta > 1$, admits a $(1+\eps)$-spanner of size $n(1/\eps)^{O(d)}$, and of pathwidth $O(n^{1-1/\delta} \log n)$. 
We show the following lower bound, which establishes that the upper bound from \cite{sidiropoulos2017algorithmic} is essentially best-possible.

\begin{theorem} \label{thm:spannerhigherdim}
Let $d \geq 2$ be an integer. Then for all $\delta \in (1,d)$, for all $\eps > 0$ and for all $n_{0} \in \mathbb{N}$, there exists a set of $n \geq n_{0}$ points $P \subset \mathbb{R}^{d}$, of fractal dimension at most $\delta'$, where $ \big| \delta - \delta^{\prime} \big| \leq \eps$, such that for any $c \geq 1$, any $c$-spanner $G$ of $P$ has $tw(G) = \Omega \left(\frac{n^{1-1/\delta^{\prime}}}{c^{d-1}} \right)$. 
\end{theorem}


\paragraph*{Independent Set of Unit Balls}
We consider the $k$-Independent Set of Unit Balls in $\mathbb{R}^d$, which is a prototypical geometric optimization problem, parameterized by the optimum.
In this problem given a set of $n$ unit balls in $\mathbb{R}^d$, we seek to find a set of $k$ pairwise non-intersecting balls.
It is known that this problem can be solved in time $n^{O(k^{1-1/d})}$, for any $d\geq 2$ \cite{alber2002geometric,MS14}, and that there is no algorithm with running time $f(k) n^{o(k^{1-1/d})}$, for any computable function $f$, assuming ETH \cite{MS14} (see also \cite{marx2005efficient}).
The upper bound has been generalized for fractal dimension as follows:
It has been shown that when the set of centers of the balls has fractal dimension $\delta$, the problem can be solved in time $n^{O(k^{1-1/\delta})+\log n}$ \cite{sidiropoulos2017algorithmic}.
We show the following lower bound on the running time, which nearly matches tis upper bound, up to a logarithmic term.

\begin{theorem}\label{thm:independent_balls}
Let $d\geq 2$ be an integer, and 
let $1<\delta'<\delta<d$. If for all $k$, and for some computable function $f$, there exists an $f(k)n^{o(k^{1-1/\delta'})}$ time algorithm for finding $k$ pairwise non-intersecting open balls in a collection of $n$ unit balls with the centers having fractal dimension at most $\delta$, then ETH fails.
\end{theorem}

\paragraph*{Euclidean TSP}

It is known that TSP on a set of $n$ points in $d$-dimensional Euclidean space can be solved in time $2^{1-1/d} n^{O(1)}$ \cite{smith1998geometric},
and that there is no algorithm with running time $2^{O(n^{1-1/d-\eps})}$, for any $\eps>0$, assuming ETH.
The upper bound has been generalized to the case of fractal dimension as follows.
It has been shown that for set of fractal dimension $\delta>1$, in $O(1)$-dimensional Euclidean space, TSP can be solved in time $2^{O(n^{1-1/\delta}\log n)}$.
Here, we obtain the following nearly-tight lower bound.

\begin{theorem}\label{thm:tsplowerboundmain}
Let $d \geq 2$ be an integer. Then, for all $\delta \in (2,d)$, for all $\delta'<\delta$ and for all $n_0 \in \mathbb{N}$, if there exists $n \geq n_{0}$ such that Euclidean TSP in $\mathbb{R}^d$ on all pointsets of size $n$ and fractal dimension at most $\delta$ can be solved in time $2^{O(n^{1-1/\delta'})}$, then ETH fails.
\end{theorem}


\subsection{Overview of techniques}

We now briefly highlight the main technical tools used in the paper.

\paragraph*{High-level idea}

We derive our lower bounds by adapting a method from the case of graph problems.
It is known that, for many problems on graphs, large treewidth implies large running time lower bounds (see, e.g.~\cite{marx2010can}).
This is typically done by exploiting the duality between treewidth and grid minors.
Specifically, it is known that, for any $r$, graphs of treewidth at least $f(r)$, for some function $r$, have the $(r\times r)$-grid as a minor \cite{robertson1986graph}.
In fact, the function $f$ is known to be linear for planar graphs \cite{robertson1994quickly}, and polynomial in general \cite{chekuri2016polynomial}.
One can often obtain a lower bound on the running time by using the grid minor to embed a large hard instance in the input.
We apply the above approach to the geometric setting by relating fractal dimension to treewidth.
Specifically, we construct a pointset in Euclidean space, such that any $O(1)$-spanner must have large treewidth.

\paragraph*{From fractal dimension to treewidth}

A main technical ingredient for obtaining nearly-optimal lower bounds is constructing pointsets such that the treewidth of any $O(1)$-spanner is as large as possible.
For the case of exposition, we will describe the construction in the continuous case.
We construct some $X\subset \mathbb{R}^d$, and we discretize $Z$ by taking some $O(1)$-approximate\footnote{Recall that a $O(1)$-approximate $r$-net in some metric space $(X,\rho)$ is some $N\subseteq X$, such that for all $x\neq y\in N$, $\rho(x,y)>r$, and for all $z\notin N$, $\rho(z,N) = O(r)$.} $\eps$-net $N_{\eps}$ of $X$.
Let us refer to the fractal dimension of the resulting infinite family of nets $N_\eps$, as the fractal dimension of $X$ (see Section \ref{subsec:preliminaries} for precise definitions).
Our goal is to construct some $X$, with some fixed fractal dimension $\delta \in (1,d]$, such that the treewidth of any $O(1)$-spanner of $N_\eps$ is as large as possible as a function of $1/\eps$.

\paragraph*{A first failed attempt: The \siep~carpet}

Let us now briefly describe the construction and point out the main technical challenges.
A natural first attempt in $\mathbb{R}^2$ is to let $X$ be the \siep~carpet. This is a set obtained from the unit square by removing the central square of side length $1/3$, and by recursing on the remaining $8$ sub-squares (see Figure \ref{fig:siep}).
Unfortunately, this construction does not lead to a tight treewidth lower bound.
Specifically, the resulting set has fractal dimension $\delta=\log 8/\log 3$, while there exist $O(1)$-spanners of treewidth $O(n^{1-1/\gamma})$, where $\gamma$ is a constant arbitrarily close to $\log 6 / \log 3$.

\begin{figure}
\begin{center}
\scalebox{0.7}{\includegraphics{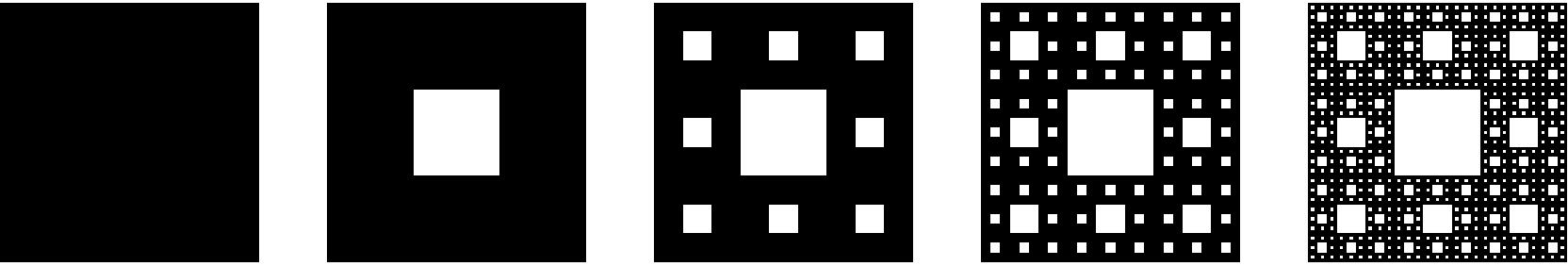}}
\caption{The first few iterations of the \siep~carpet.\label{fig:siep}}
\end{center}
\end{figure}

Intuitively, this happens for the following reason.
Let $S_\eps$ be any $O(1)$-spanner for $N_\eps$.
Then, the largest possible grid minor in $S_\eps$ does not use most of the vertices in $S_\eps$.
Thus, roughly speaking, we can obtain a larger grid minor by constructing a set $X$ so that as few vertices of $S_\eps$ as possible are being ``wasted''.

\paragraph*{Constructing a treewidth-extremal fractal: The Cantor crossbar}

Using the above observation, we define the set $X$ as follows.
We first recall that the Cantor set ${\cal C}$ is obtained from the unit interval by removing the central interval of length $1/3$, and recursing on the other two (see Figure \ref{fig:cantor}).
We define ${\cal C}'$ to be the Cartesian product of ${\cal C}$ with $[0,1]$, and we set $X$ to be the union of two copies of ${\cal C}'$, where one is rotated by $\pi/2$.
We refer to the resulting set as the \emph{Cantor crossbar} (see Figure \ref{fig:cantor_crossbar_R2}).
We can sow that the resulting set achieves a nearly-optimal treewidth lower bound.

\begin{figure}
\begin{center}
\scalebox{0.6}{\includegraphics{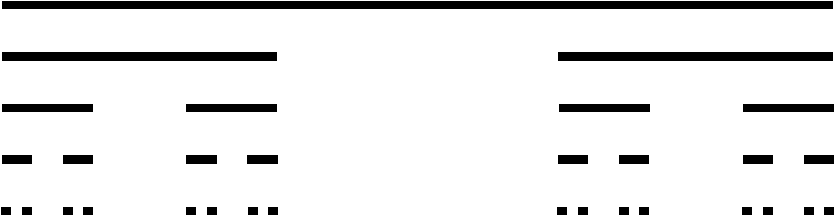}}
\caption{The first few iterations of the Cantor ternary set.\label{fig:cantor}}
\end{center}
\end{figure}

\begin{figure}
\begin{center}
\scalebox{0.7}{\includegraphics{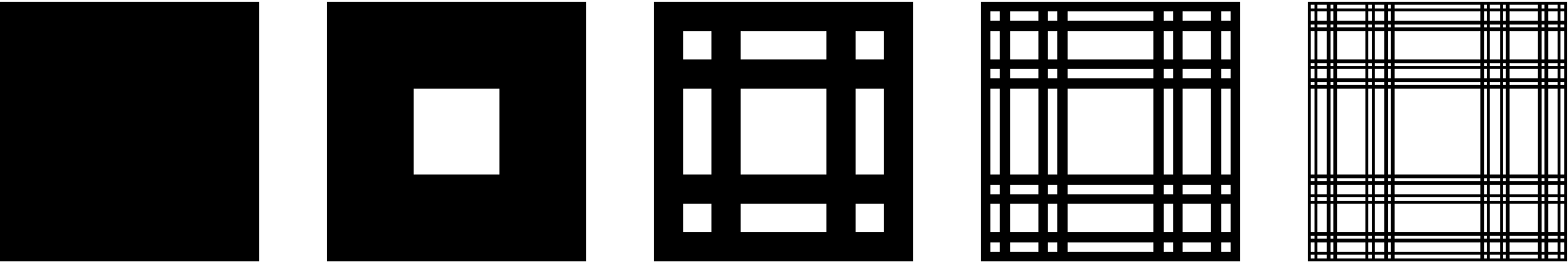}}\\
\caption{The first few iterations of the Cantor crossbar in $\mathbb{R}^2$.\label{fig:cantor_crossbar_R2}}
\end{center}
\end{figure}

The above construction can be generalized to the case where the ambient dimension is $d\geq 2$ as follows.
Recall that, for any $d'\geq 1$, the \emph{Cantor dust} in $\mathbb{R}^{d'}$, denoted by ${\cal D}_{d'}$, is the Cartesian product of $d'$ copies of the Cantor set (see Figure \ref{fig:cantor_dust}).
\begin{figure}
\begin{center}
\scalebox{0.7}{\includegraphics{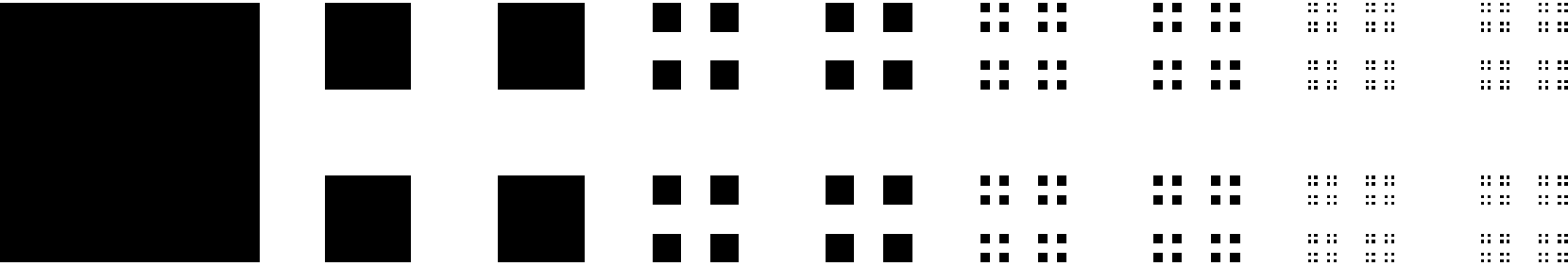}}\\
\caption{The first few iterations of the Cantor dust in $\mathbb{R}^2$.\label{fig:cantor_dust}}
\end{center}
\end{figure}
Let $e_1,\ldots,e_d$ be the standard orthonormal basis in $\mathbb{R}^d$.
For each $i\in \{1,\ldots,d\}$, we define ${\cal T}_i$ to be the Cartesian product of ${\cal D}_{d-1}$ with $[0,1]$, rotated so that $[0,1]$ is parallel to $e_i$.
Finally, we set $X={\cal T}_1\cup \ldots \cup {\cal T}_d$ (see Figure \ref{fig:cantor_crossbar_R3}).

\begin{figure}
\begin{center}
\scalebox{0.18}{\includegraphics{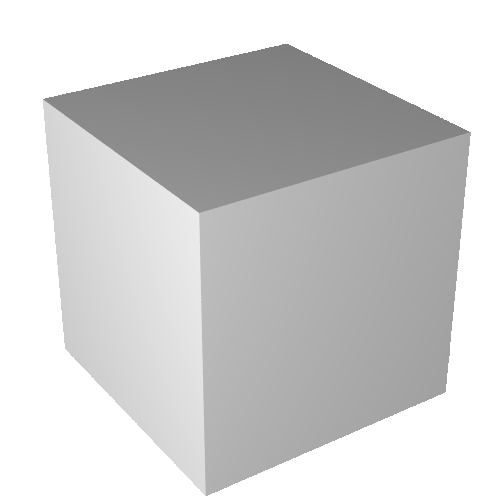}}
\scalebox{0.18}{\includegraphics{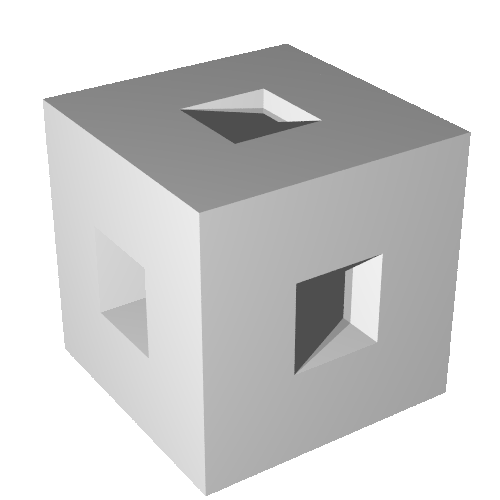}}
\scalebox{0.18}{\includegraphics{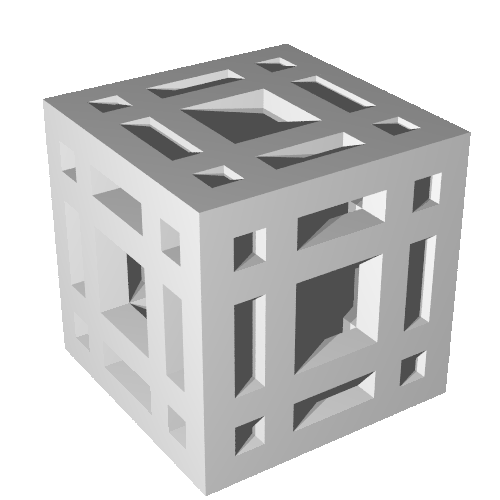}}
\scalebox{0.18}{\includegraphics{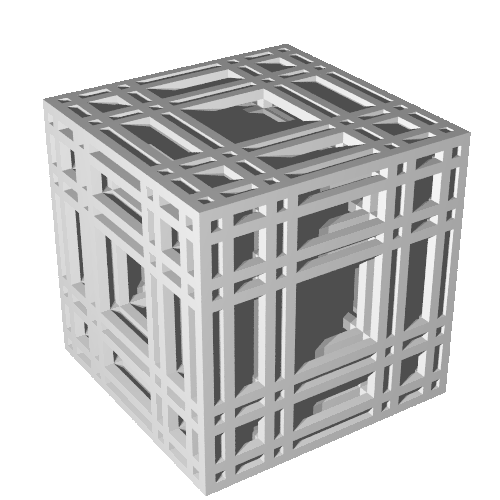}}
\caption{The first few iterations of the Cantor crossbar in $\mathbb{R}^3$.\label{fig:cantor_crossbar_R3}}
\end{center}
\end{figure}

The above construction gives a set with some fixed fractal dimension $\delta$, for each fixed $d\geq 2$.
We can generalize the construction so that $\delta$ attains any desired value in the range $(1,d]$.
The only difference is that, when defining the Cantor dust, we start with a Cantor set of smaller dimension.
This can be done by removing the central interval of length $\alpha\in (0,1)$, instead of $1/3$, and recursing on the remaining two intervals of length $(1-\alpha)/2$.

\paragraph*{From spanner lower bounds to running time lower bounds}

The above construction of the Cantor crossbar leads to a nearly-optimal lower bound for the treewidth of $O(1)$-spanners.
We next use this construction to obtain running time lower bounds.
Informally, a typical NP-hardness reduction for some geometric problem in the plane works as follows:
One encodes some known computationally hard problem by constructing ``gadgets'' that are arranged in a grid-like fashion in $\mathbb{R}^2$ (see, e.g.~\cite{MS14}).
More generally, for problems in $\mathbb{R}^d$, the gadgets are arranged along some $d$-dimensional grid.
We follow a similar approach, with the main difference being that we arrange the gadgets along a Cantor crossbar.

\subsection{Other related work}

There has been a large body of work on determining the effect of \emph{doubling} dimension on the complexity of various geometric problems \cite{har2006fast,bartal2012traveling,cole2006searching,krauthgamer2005black,gottlieb2008optimal,krauthgamer2005measured,chan2009small,chan2005hierarchical,gupta2012online,talwar2004bypassing}.
Other notions of dimension that have been considered include low-dimensional negatively curved spaces \cite{krauthgamer2006algorithms}, 
growth-restricted metrics \cite{karger2002finding},
as well as 
generalizations of doubling dimension to metrics of so-called bounded global growth \cite{hubert2012approximating}.
In all of the above lines of research the goal is to extend tools and ideas from the Euclidean setting to more general geometries.
In contrast, we study restricted classes of Euclidean instances, with the goal of obtaining better bounds than what is possible in the general case.



\subsection{Preliminaries} \label{subsec:preliminaries}
We give some definitions that are used throughout the paper.
\begin{definition}[\textbf{Fractal dimension}] \cite{sidiropoulos2017algorithmic}
The \emph{fractal dimension} of some $P\subseteq \mathbb{R}^d$, denoted by $\dimF(P)$, is defined as the infimum $\delta$, such that for any $\eps>0$ and $r\geq 2\eps$, for any $\eps$-net $N$ of $P$, and for any $x\in \mathbb{R}^d$, we have 
$|N \cap \ball(x,r)| = O((r/\eps)^\delta)$.
\end{definition}

We have the following lemmas showing invariance of fractal dimension under certain operations.

\begin{lemma}\label{lem:scaling}
Let $d \geq 1$ be an integer, let $0 < \delta \leq d$ and let $c>0$ be some constant. Let $P \subset \mathbb{R}^d$ be a pointset such that $\dimF(P) = \delta$. Let $P'$ be the pointset obtained by uniformly scaling the points of $P$ about the origin by a factor of $c$. Then $\dimF(P') = \delta$.    
\end{lemma}

\begin{proof}
Fix a constant $c > 0$. Consider the natural mapping $f : P \to P'$ where every point $p$ maps to the point obtained by scaling $p$ by a factor of $c$ about the origin. For all sets $X \subseteq P$ we denote by $f(X)$ the set $Y = \{f(p)~|~p \in X \}$. Similarly we use $f^{-1}(Y)$ to refer to the inverse mapping to $X$. First we observe that that for all $\eps>0$, any $\eps$-net of $P$ maps to a $c \eps$-net of $P'$ under $f$ and any $\eps$-net of $P'$ maps to an $\frac{\eps}{c}$-net of $P$ under $f^{-1}$, since all pairwise distances are scaled by exactly a factor of $c$ in $P'$. Similarly for any $x \in \mathbb{R}^d$ and $r>0$, we have that $P' \cap \ball(x,r)$ is mapped to by $P \cap \ball \left(f^{-1}(x),\frac{r}{c} \right)$. Therefore for any $\eps>0$ and $r \geq 2\eps$, for any $\eps$-net $N$ of $P'$, and for any $x \in \mathbb{R}^d$, we have $|N \cap \ball(x,r)| = \left|f^{-1}(N) \cap \ball \left(f^{-1}(x),\frac{r}{c} \right) \right| = O \left( \left(\frac{r/c}{\eps /c} \right)^{\delta} \right) = O \left( \left(\frac{r}{\eps} \right)^\delta \right)$. From the definition of fractal dimension, it follows that $\dimF(P') = \delta$. 
\end{proof}

\begin{lemma} \label{lem:substitution}
Let $d \geq 1$ be some integer, let $0 < \delta \leq d$ and let $c > 0,k>0$ be some constants. Let $P \subset \mathbb{R}^d$ be a pointset such that $\dimF(P) = \delta$ and for all $u,v \in P$, $d(u,v) > 4c $. For all $p \in P$ let $S_p \subset \mathbb{R}^d$ be a set of points such that $|S_p| \leq k$ and for all $x \in S_p$, $d(x,p) \leq c$. Let $P' = \bigcup\limits_{p \in P} S_p$. Then $\dimF(P') = \delta$.    
\end{lemma}

\begin{proof}
Fix constants $c>0$ and $k>0$. For all $p \in P$, let $\rep(p) \in S_p$ denote an arbitrarily chosen representative point. Let $\eps >0$ and let $N$ be an $\eps$-net of $P'$. 

First we consider the case when $\eps > 4c$. Let $M$ be some $ \left(\frac{\eps}{2} -2c \right)$-net of $P$. Consider the pointset $M' = \{ \rep(p)~|~p \in M \} $. Then $M'$ is an $\frac{\eps}{2}$-cover of $P'$. Since $N$ is $\eps$-packing, this implies that every point in $N$ is covered by a unique point in $M'$. Therefore for any $\eps>0$ and $r \geq 2\eps$, for any $\eps$-net $N$ of $P'$, and for any $x \in \mathbb{R}^d$, we have $|N \cap \ball(x,r)| \leq |M' \cap \ball(x,r+ \frac{\eps}{2})| \leq |M \cap \ball(x,r+ \frac{\eps}{2} + c)|$. Since the pointset $P$ has fractal dimension $\delta$, this implies that $|N \cap \ball(x,r)| \leq O\left( \left(\frac{r+ \frac{\eps}{2} + c}{\frac{\eps}{2} -2c} \right)^{\delta} \right) = O\left(\left(\frac{r}{\eps}\right)^{\delta}\right)$. From the definition of fractal dimension this implies that $\dimF(P') = \delta$.

Suppose instead that $\eps \leq 4c$. Then let $M$ be some $\eps$-net of $P$. Since for all $u,v \in P$ $d(u,v) > 4c$, it must be that $M = P$. Let $M' = P'$. Now using the same argument as before, we have that for any $\eps>0$ and $r \geq 2\eps$, for any $\eps$-net $N$ of $P'$, and for any $x \in \mathbb{R}^d$, $|N \cap \ball(x,r)| \leq \left|M' \cap \ball \left(x,r+ \frac{\eps}{2} \right) \right|$. Since $M' = P'$ and $M = P$, it follows that $\left|M' \cap \ball \left(x,r+ \frac{\eps}{2} \right) \right| \leq k \cdot \left|M \cap \ball \left(x,r+ \frac{\eps}{2} +c \right) \right|$. Now since $M$ is an $\eps$-net of $P$ $k \cdot \left|M \cap \ball \left(x,r+ \frac{\eps}{2} +c \right) \right| \leq k \cdot O \left(\frac{r+ \frac{\eps}{2} +c}{\eps} \right)^{\delta} = O \left(\frac{r}{\eps} \right)^{\delta}$. Again, from the definition of fractal dimension, it follows that $\dimF(P') = \delta$.
\end{proof}

\begin{definition}[\textbf{$c$-spanner}]
For any pointset $P \subset \mathbf{R}^d$, and for any $c \geq 1$, a \emph{$c$-spanner} for $P$ is a graph $G$ with $V(G) = P$, such that for all $x,y \in P$, we have
$$ ||x-y||_2 \leq d_G(x,y) \leq c \cdot ||x-y||_2, $$
where $d_G$ denotes the shortest path distance in $G$.
\end{definition}

\begin{definition}[\textbf{Treewidth}] \cite{Diestel12}
Let $G$ be a graph, $T$ a tree and $\mathcal{V} = \{ V_t \}_{t \in T}$ be a family of vertex sets $V_t \subseteq V(G)$ indexed by the vertices $t$ of $T$. The pair $(T, \mathcal{V})$ is called a tree-decomposition of $G$ if it satisfies the following three conditions:
\begin{enumerate}
\item $V(G) = \cup_{t \in T} V_t$.
\item For every edge $e \in G$, there exists a $t \in T$ such that both ends of $e$ lie in $V_t$.
\item $V_{t_1} \cap V_{t_3} \subseteq V_{t_2}$, whenever $ t_2$ lies in the unique path joining $t_1$ and $t_3$ in $T$.
\end{enumerate}
The width of $(T,\mathcal{V})$ is the number $\max \{|V_t|~:~t \in T \}$ and the \emph{treewidth} of $G$ is the least width of any tree-decomposition of $G$.
\end{definition}

\subsection{Organization}
This paper is organized as follows. Section \ref{sec:tw-lower-bound} presents lower bound on the treewidth of spanners for arbitrary pointsets with integral dimension, and with fractal dimension. Section \ref{sec:independent-balls} presents running time lower bound on the Independent Set of Balls problem on pointsets with arbitrary fractal dimension in $\mathbf{R}^d$. The proof of Theorem \ref{thm:tsplowerboundmain} has been moved to Section \ref{sec:TSP} of the appendix. This theorem proves running time lower bound on the Euclidean TSP problem on pointsets with arbitrary fractal dimension in $\mathbf{R}^d$. The case of $d=2$ is discussed in Section \ref{subsec:lowerboundtsp} and the case of $d \geq 3$ is discussed in Section \ref{subsec:lowerboundtsparbitdim}.

\section{Lower bounds on the treewidth of spanners} \label{sec:tw-lower-bound}
In this section, we obtain lower bounds on the treewidth of spanners for arbitrary pointsets. In subsection \ref{subsec:twintegerdim}, we consider pointsets with integral fractal dimension. In subsection \ref{subsec:sierp}, we consider a discretized version of the \siep~carpet whose fractal dimension is less than two but greater than 1. In subsection \ref{subsec:modifiedsierp}, we use a carefully chosen inductive construction to obtain a specific fractal pointset of fractal dimension $\frac{\log 6}{ \log 3}$. This pointset gives us a nearly tight lower bound on the treewidth of a spanner. We finally generalize this construction in subsection \ref{subsec:twfractal} and present the proof of Theorem \ref{thm:spannerhigherdim}. 

\subsection{Treewidth and integral dimension}\label{subsec:twintegerdim}
We obtain lower bounds on the treewidth of spanners for pointsets with integral fractal dimension. We will make use of the following Threorem due to Kozawa \etal \cite{Kozawaetal10} for the proofs in this section.
\begin{theorem}\label{thm:gridlowerbound} \textnormal{\cite{Kozawaetal10}}
The treewidth of the $d$-dimensional grid on $n$ vertices is $\Theta(n^{1-1/d})$.
\end{theorem}

\begin{theorem}\label{thm:ddimensionalgridspanner}
For any integer $d \geq 1$, there exists a set of $n$ points $P \subseteq \mathbb{R}^{d}$ such that for any $c\geq 1$, and for any $c$-spanner $G$ of $P$, $\tw(G) = \Omega \left( \frac{n^{1-1/d}}{c^{d-1}} \right)$. 
\end{theorem}
\begin{proof}
Fix an integer $d \geq 1$. Let $n \in \mathbf{N}$ be such that $n^{\frac{1}{d}}$ is an integer. Let $P = \{ 1,2, \cdots, n^{\frac{1}{d}} \}^{d}$ and let $G$ be any $c$-spanner of $P$. Let $p = (p_1, \ldots, p_d)$ be a point in $P$. We define $P'$ and $X'$ as follows:
$$P' = \{ p \in P~|~ \exists~ 1 \leq i \leq d \text{ such that}~\forall~ j \neq i, ~(p_j ~\mbox{mod}~ (c+1)) = 1 \}.$$
$$X = \{ p \in P'~|~ \forall~ 1 \leq i \leq d,~  ( p_i ~\mbox{mod}~ (c+1)) = 1 \}.$$ 
Consider the points in $P'$. We call a row of points in $P'$ that is parallel to one of the $d$ axes a \emph{full row} if this row consists of exactly $n^{\frac{1}{d}}$ points with adjacent points unit distance apart. Consider for any $1 \leq i \leq d$, the points of any pair of full rows $R$ and $T$ that are both parallel to the $i$\textsuperscript{th} axis. We have that for any pair of consecutive points $x_{1}, x_{2}$ in $R$ and for any pair of consecutive points $y_{1}, y_{2}$ in $T$, no shortest path in $G$ joining $x_{1}$ and $x_{2}$ can intersect any shortest path in $G$ joining $y_{1}$ and $y_{2}$. Suppose not, then let $z$ be a point of intersection between two such shortest paths. Since $G$ is a $c$-spanner and consecutive points in any row are distance $1$ apart, we have $d_G(x_{1} , z) + d_G(x_{2} , z) \leq c$ and similarly, $d_G(y_{1} , z) + d_G(y_{2} , z) \leq c$. But this implies that at least one of $d_G(x_{1},y_{1})$, $d_G(x_{1},y_{2})$, $d_G(x_{2},y_{1})$ and $d_G(x_{2},y_{2})$ is at most $c$ due to triangle inequality. This is a contradiction because $G$ is non-contracting and the distance between $R$ and $T$ is at least $c+1$ by our choice of $R$ and $T$.

Now if we consider a shortest path between every pair of consecutive points in the row $R$ and concatenate these paths, remove all loops, then we can obtain a path from one end of this row to the other end. Doing the same for all full rows in $P'$, we end up with a set of paths traversing the points in the full rows. Moreover from the earlier argument, it follows that any two such paths obtained from parallel full rows are vertex disjoint. Thus for all $1 \leq i \leq d$, we can obtain a set of vertex disjoint paths $Q_i$ in $G$ that traverse the points in the full rows of $P'$ parallel to the $i$\textsuperscript{th} axis. 

Finally we define a subgraph $H$ of $G$ as follows: $H$ consists of the points in $X$. Furthermore, for any pair of points $p,q \in X$ such that $p$ and $q$ differ only along one coordinate, say the $i$\textsuperscript{th} coordinate, and differ by exactly $c+1$, $H$ also consists of the sub-path between $p$ and $q$ of the corresponding path in $Q_i$ connecting $p$ and $q$. Now contracting these paths in $H$ between adjacent points in $X$ results in a $d$-dimensional grid with $\frac{n}{\lceil (c+1) \rceil^{d}}$ points. Thus, we conclude that $\tw(G) = \Omega \big( \frac{n^{1-1/d}}{(c+1)^{d-1}} \big) = \Omega \big( \frac{n^{1-1/d}}{c^{d-1}} \big)$.  
\end{proof}

\subsection{A first attempt: The \siep~carpet}\label{subsec:sierp}

Consider a set of points $X$ obtained by the following method: start with a $3^{k} \times 3^{k}$ integer grid for some $k \in \mathbb{N}$ and partition it into $9$ subgrids of equal size. We delete all the points in the central subgrid and recurse on the remaining $8$ subgrids. The recursion stops when we arrive at a subgrid containing a single point. This is a natural discrete variant of the \siep~ carpet.
\begin{theorem}\label{thm:fractaldimension}
Let $\delta$ be the fractal dimension of the set of points $X$ obtained above. Then we have $\delta \geq \log_{3} 8$.
\end{theorem}
\begin{proof}
We start by recalling the definition of fractal dimension of a set of points $P$. It is the infimum $\delta > 0$ such that for any $\eps > 0$, for any ball $B$ of radius $r \geq 2 \eps$ and for any $\eps$-net $N$, we have $ \vert B \cap N \vert = O((r/ \eps)^{\delta})$. 

As seen in the construction, the width of the grid reduces by $\frac{1}{3}$ at every step of recursion. Thus we may assume that the width is $\frac{1}{3^{i}}$ when we stop. Let $\eps = \frac{1}{2} \cdot \frac{1}{3^{i}}$. Let $N = X$. We have $n = \vert X \vert = 8^{i}$ since every step of recursion is done on the remaining $8$ subgrids. Let $r = \sqrt{2}$. Then $\vert B(x,r) \cap N \vert = n = 8^{i} \leq m \cdot \big( \frac{r}{\eps} \big) ^{\delta}$ for some constant $m > 0$. This gives $8^{i} \leq m \cdot (\sqrt{2} \cdot 2 \cdot 3^{i})^{\delta}$. Taking $\log$ on both sides, we get $ \log 8 \leq \delta \log 3$. Thus, $ \delta \geq \log_{3} 8$.
\end{proof}

\begin{theorem} \label{thm:sierpinski-carpet}
There exists a set of $n$ points $P$ with fractal dimension $\delta \in (1,2)$, and some $c$-spanner $G$ of $P$, where $c < 1+ \sqrt{2}$, such that $\tw(G) = \Theta(n^{1-\frac{1}{\delta}-\eps})$, for some $\eps>0$.
\end{theorem}

\begin{proof}
We consider the set of points $P$ of the \siep~carpet. For completion, we briefly describe the construction below:

We start with a $[0,3] \times [0,3]$ integer grid. This is a grid with $9$ equal sized boxes. We leave the central box and partition each of the remaining 8 boxes into a $3 \times 3$ grid. We keep recursing on the new boxes obtained, every time leaving the central box. In the end, every box contains a single point. The pointset obtained in the end is a $3^{k} \times 3^{k}$ grid with holes and a point in each of the boxes of the grid. The set of vertices $P$ of this grid form a \siep~carpet.

We now add edges between grid points of $P$ such that the graph obtained is a subgraph of the $3^{k} \times 3^k$ grid graph. Now the points inside the boxes are joined to the grid graph in the following manner: for every box, the point inside it is joined to one of the pairs of diagonally opposite corners of the box. We observe that this is a $c$-spanner with $c < 1+ \sqrt{2}$. We denote this $c$-spanner by $G$. 

Consider the subgraph $H$ of $G$ obtained by removing the points inside the boxes of the $3^{k} \times 3^{k}$ grid and the edges adjacent to them. We now show that $H$ contains a $n^{\frac{1}{3}} \times n^{\frac{1}{3}}$ grid as a minor. The minor is obtained as follows: consider the central boxes (holes) in the grid that are left every time with no further partitioning. We compress each of these boxes to a box of the size of those in the $3^{k} \times 3^{k}$ grid. We denote the resulting graph by $M$.

We claim that $M$ is an $n^{\frac{1}{3}} \times n^{ \frac{1}{3}}$ grid. We observe that after $k$ steps of recursion, the minimum cut in $G$ has $2^{k}$ edges. This is because the minimum cut in $G$ is the one that partitions $G$ into two equal parts (either horizontally or vertically) and passes through the largest empty central box in $G$. Since the size of $M$ is same as the size of the minimum cut in $G$, we get that $M$ is a $2^{k} \times 2^k$ grid. The number of points $n = 8^{k}$. This gives $k = \log_{8} n$ and $2^{k} = n^{\frac{1}{3}}$.

From the graph minor theorem, we have $\tw(G) \geq \tw(H) \geq \tw(M)$. Thus $\tw(G) = \Omega(n^{\frac{1}{3}})$. For $\eps \geq \frac{2- \log 3}{3}$, we get $$n^{\frac{1}{3}} \geq n^{1- \frac{\log 3}{3} - \eps} = n^{1 - \frac{1}{\delta} - \eps}$$ 
Here, $\delta = \log_{3} 8$ from Theorem \ref{thm:fractaldimension}. This gives $\tw(G) = \Omega(n^{1- \frac{1}{\delta} - \eps})$ for some fixed constant $\eps > 0$. 

This establishes a lower bound on $\tw(G)$. We now prove a similar upper bound. This is done by constructing a tree decomposition of the spanner $G$. The construction is done in the following manner:

The root bubble consists of vertices in the two middle rows and two middle columns of the graph $G$. The set of edges between the two middle columns form a minimum cut of $G$. The same is true for the set of edges between the two middle rows. From above, the size of minimum cut in $G$ is $2^{k}$. Thus the number of vertices in the root bubble is $\Theta(2^{k})$. 

The removal of vertices of the root bubble partitions $G$ into $4$ connected components. Let the connected components be $G_{1}, G_{2}, G_{3}$ and $G_{4}$. We now look at the minimum cut for each of these components. As before, we look at the minimum cut in the horizontal direction as well as in the vertical direction. For $i \in [4]$, the minimum cut in $G_{i}$ along both horizontal and vertical directions passes through the largest empty box in $G_{i}$. Note that for all $i \in [4]$, $G_{i}$ contains a copy of the \siep~carpet obtained after $k-1$ levels of recursion. Therefore we get that the size of the minimum cut in $G_{i}$ along both horizontal and vertical directions is $\Theta(2^{k-1})$. For $i \in [4]$, let the minimum cut (in the horizontal as well as vertical direction) in $G_{i}$  be denoted by $S_{i}$ respectively. The root bubble has four children and the $i^{\mathrm{th}} $ child consists of vertices of the root bubble as well as vertices of the set $S_{i}$.

In order to construct the third level of the tree, we look at the graph obtained after removing vertices in the bubbles at the second level. In each of the connected components obtained, we again look at the minimum cut in  horizontal as well as vertical direction. The size of the minimum cut now is $\Theta(2^{k-2})$. Every bubble in the second level has $4$ children. Each of these children consists of vertices corresponding to the minimum cut of a connected component union vertices of all the bubbles in its ancestor path (i.~e.~bubbles occurring in the shortest path to the root bubble). 

We keep constructing more levels of our tree in a similar manner as above. We stop when all edges of $G$ get covered. The tree thus obtained has $k$ levels. Every bubble at level $l$ contains $\Theta(2^{k} + 2^{k-1} + \cdots + 2^{k-l+1})$ vertices.  

We now show that the decomposition obtained above is a valid tree decomposition. By construction, all edges are covered in this decomposition. Every bubble that is not a leaf has $4$ children and contains all vertices belonging to  the bubbles in its ancestor path. Thus, all conditions of a tree decomposition are satisfied. 

The treewidth of this decomposition is $\Theta(2^{k})$. This shows that $\tw(G) = O(2^{k})$. As shown above, $2^{k} = n^{\frac{1}{3}}$ and for $\eps = \frac{2 - \log 3}{3}$, $n^{\frac{1}{3}} = n^{1- \frac{1}{\delta} - \eps}$. Thus combining the lower and upper bounds, we get that $\tw(G) = \Theta(n^{1 - \frac{1}{\delta} - \eps})$.
\end{proof}



\subsection{One treewidth-extremal fractal: A discretized Cantor crossbar}\label{subsec:modifiedsierp}

In this section we describe the construction of a pointset in $\mathbb{R}^2$ with a specific fractal dimension of $\frac{\log 6}{\log 3}$. This pointset is a discretized version of the \emph{Cantor crossbar}. We generalize this construction in the next section to construct pointsets with arbitrary fractal dimension.

\begin{theorem}\label{thm:modifiedsieplowerbound}
Let $\delta = \frac{\log(6)}{\log(3)}$. Then for all $n_0 \in \mathbb{N}$, there exists a set $P\subset \mathbb{R}^2$ of $n \geq n_0$ points of fractal dimension at most $\delta$, such that for any $c \geq 1$, any $c$-spanner $G$ of $P$ has $\tw(G) = \Omega \left(\frac{n^{1-\frac{1}{\delta}}}{c} \right)$.
\end{theorem}

\paragraph{Construction of the discrete Cantor crossbar}.~To prove the above theorem, we consider the set of points obtained as follows. First we define $f(0)$ and $h(0)$ and $g(0)$ to be a single point. Then we inductively define $h(i)$,$g(i)$ and $f(i)$ as follows. To get $h(i)$ we start with a $3^i \times 3^i$ integer grid and subdivide it into nine $3^{i-1} \times 3^{i-1}$ integer grids. Then we remove all $3$ sub-grids in the middle row and replace each of the $6$ remaining sub-grids with copies of $h(i-1)$. To get $g(i)$ we again start with a $3^i \times 3^i$ integer grid and subdivide it into nine $3^{i-1} \times 3^{i-1}$ integer grids. Then we remove all $3$ sub-grids in the middle column and replace each of the $6$ remaining sub-grids with copies of $g(i-1)$. Finally to get $f(i)$ we start with a $3^i \times 3^i$ integer grid and subdivide it into nine $3^{i-1} \times 3^{i-1}$ integer grids. Then we remove the central sub-grid. We then replace the four corner sub-grids with copies of $f(i-1)$. We replace the middle sub-grid in the first and last rows with copies of $h(i-1)$ and we replace the middle sub-grids of the first and last columns with copies of $g(i-1)$ as depicted in Figure \ref{fig:fhg}. The pointset we require is given by $f(k)$ where $k$ is any positive integer. We have the following two lemmas regarding the pointset $f(k)$.

\begin{figure}
\begin{center}
\scalebox{0.07}{\includegraphics{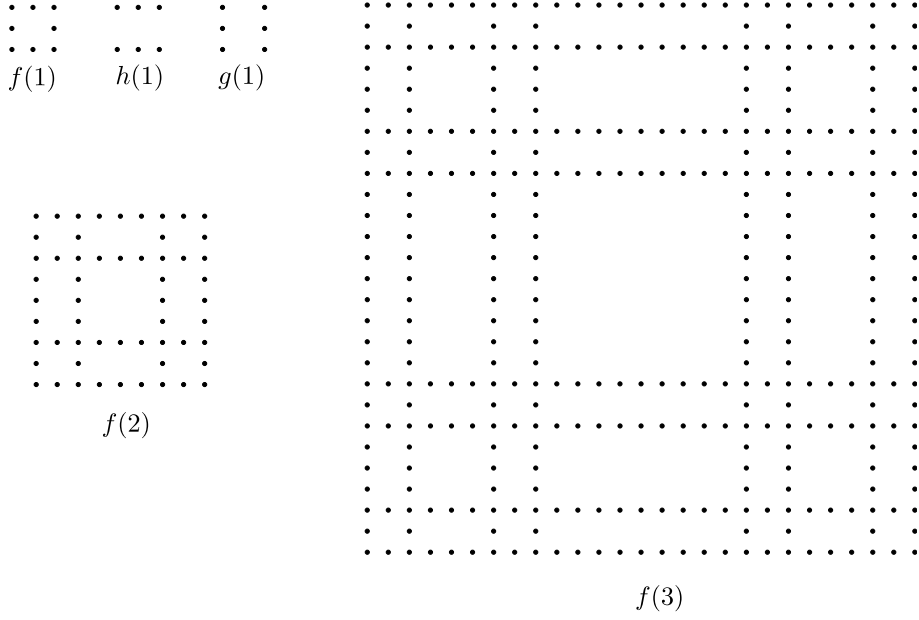}}
\caption{The construction of the discrete Cantor crossbar.}
\label{fig:fhg}
\end{center}
\end{figure}

\begin{lemma}\label{lem:cardinality_bound}
$|f(k)| \leq 2 \cdot 6^k.$
\end{lemma}

\begin{proof}
From the construction of $f(k)$ we have that $|f(k)| \leq 4|f(k-1)| + 2|h(k-1)| + 2|g(k-1)|$. By symmetry we have that $|h(k-1)| = |g(k-1)|$ and so this implies that $|f(k)| \leq 4|f(k-1)| + 4|h(k-1)|$. From the construction of $h(k)$ it follows that $|h(k)| = 6 |h(k-1)|$ which implies that $|h(k)| = 6^k$. Substituting this in the previous inequality and solving gives us that $|f(k)| \leq 2 \cdot 6^k$.
\end{proof}

\begin{lemma}\label{lem:fractaldim}
Let $P = f(k)$. Then, $\dimF(P) \leq \frac{\log(6)}{\log(3)}$.
\end{lemma}

\begin{proof}
Let $x \in \mathbb{R}^2$, $\eps > 0$ and $r \geq 2 \eps$. Let $j$ be the smallest non-negative integer such that $\eps \leq 3^j$. First we define a family of $\eps$-coverings of $P$. For all integers $i \geq 0$ we define $f'(i)$, $h'(i)$ and $g'(i)$ to be  $\eps$-coverings of the points in $f(i)$, $h(i)$ and $g(i)$ respectively. We will define them inductively similar to the construction described earlier. For all $i < j$ we define $f'(i)$, $h'(i)$ and $g'(i)$ to be a single point picked from $f(i)$, $h(i)$ and $g(i)$ respectively. Since the diameter of $f(i)$ is $\sqrt{2}3^i$ the bottom rightmost point of the top left copy of $f(i-1)$ in $f(i)$ is a point that covers $f(i)$. A similar choice can be made for $h'(i)$ and $g'(i)$. Now for the case where $i \geq j$ we may inductively obtain $f'(i)$, $h'(i)$ and $g'(i)$ in a similar fashion to the construction of $f(i)$, $h(i)$ and $g(i)$ respectively. $f'(i)$ can be obtained by starting with $f(i)$ and replacing the copies of $f(i-1)$, $h(i-1)$ and $g(i-1)$ with copies of $f'(i-1)$, $h'(i-1)$ and $g'(i-1)$ respectively. A similar approach can be used to obtain $h'(i)$ and $g'(i)$. 

Now consider $\ball(x,r)$ and let $N$ be a $2 \eps$-net of $P$ and $C$ be an $\eps$-covering of $P$. We have that $|N \cap \ball(x,r)| \leq |C \cap \ball(x,r + \eps)|$. This is because every point in $N$ is covered by a unique point in $C$ since $C$ is an $\eps$-covering and $N$ is $2 \eps$-packing. Now we may set $C$ to be $f'(k)$ which is an $\eps$-covering of $P = f(k)$. Let $t \geq 0$ be the smallest integer such that $3^t \geq 2r$. Now we have that $|f'(k) \cap \ball(x,r + \eps)| \leq 5 |f'(t)|$ since the diameter of $\ball(x,r + \eps)$ is at most $3r$. Thus we have that $|N \cap \ball(x,r)| \leq  5 |f'(t)|$. Using an argument similar to that of lemma \ref{lem:cardinality_bound} it follows that $5 |f'(t)| \leq 5 |f'(j-1)| 2 \cdot 6^{t-j +1} \leq 10 \left(\frac{3^{j}}{\eps} \right)^{\frac{\log{(6)}}{\log{(3)}}} \cdot 3^{(t-j +1) \left(\frac{\log{(6)}}{\log{(3)}} \right)} \leq 10 \left(\frac{3^{t+1}}{\eps} \right)^{\frac{\log{(6)}}{\log{(3)}}} \leq 60 \left(\frac{3r}{\eps} \right)^{\frac{\log{(6)}}{\log{(3)}}} \leq 558 \left(\frac{r}{2 \eps} \right)^{\frac{\log{(6)}}{\log{(3)}}}$. But from the definition of fractal dimension and the fact that we have $|N \cap \ball(x,r)| = O\left( \left(\frac{r}{2\eps} \right)^{\delta} \right)$, it follows that $\dimF(P) \leq \frac{\log(6)}{\log(3)}$. 
\end{proof}

\begin{lemma}\label{lem:modified-sierp}
Let $G$ be a $c$-spanner of $P$, where $c \geq 1$. Then, $\tw(G) = \Omega \left( \frac{2^{k}}{c} \right)$.
\end{lemma} 

\begin{proof}
We observe that $f(k)$ has $2^{k}$ rows with the property that any two consecutive points in the same row are at distance $1$ from each other. This is because the number of such rows in $h(k-1)$ is $2^{k-1}$. We refer to such a row as \textit{full row}. Similarly, the number of columns in $f(k)$ with the property that any two consecutive points in the same column are at distance $1$ from each other is $2^{k}$. We refer to such a column as \textit{full column}. 

Now we partition the rows of $f(k)$ into $\frac{2^{k}}{c}$ sets with each set containing $c$ consecutive full rows. Precisely, the first set consists of the first $c$ full rows and all other rows in between them, the second set starts from the next  uncovered row and extends till $c$ full rows are covered and so on. In the end, we get a grouping of rows of $f(k)$ with each set (maybe except the last one) containing $c$ consecutive full rows. We label these sets as $S_{1}, S_{2}, \ldots$ in the order of construction. The number of sets might be less than $ \frac{2^{k}}{c} $. We do a similar grouping for the columns and label the sets formed as $T_{1}, T_{2}, \ldots$ in the order of construction.

Consider a set $S_{i}$ and take the middle full row of $S_{i}$. For any pair of consecutive points $x_{1}, x_{2}$ in this row, the path joining $x_{1}$ and $x_{2}$ lies entirely within $S_{i}$. This is because $G$ is a $c$-spanner of $P$. By concatenating these paths and removing all loops, we can obtain a path from one end of this row to the other end that lies entirely within $S_{i}$. Let this path be $P_{S_{i}}$. Such a path exists inside every set $S_{i}$. In a similar manner, every set $T_{i}$ of columns contains a path $P_{T_{i}}$ from the top to the bottom of the column that lies entirely within $T_{i}$.

We now look at paths $ P_{S_{1}}, P_{S_{2}}, \ldots, P_{T_{1}}, P_{T_{2}}, \ldots $. Let the point of intersection of $P_{S_{i}}$ and $P_{T_{j}}$ be $x_{i,j}$. Then, the graph consisting of all $x_{i,j}$ and paths joining $x_{i,j}, x_{i,j+1}$ and $x_{i,j}, x_{i+1,j}$ for all $i,j$ is a subgraph of $G$. All these paths are disjoint (except at the end points) since they belong to different sets. Moreover, if each of these paths is contracted to an edge, then we obtain a $ \frac{2^{k}}{c} \times \frac{2^{k}}{c} $ grid. Thus, we conclude that $\tw(G) = \Omega \left( \frac{2^{k}}{c} \right)$.  
\end{proof}

\begin{proof}[Proof of Theorem \ref{thm:modifiedsieplowerbound}]
Using Lemma \ref{lem:fractaldim} and Lemma \ref{lem:modified-sierp}, we get that $\tw(G) = \Omega \left(\frac{(3^{k \delta})^{1-\frac{1}{\delta}}}{c} \right) = \Omega \left( \frac{n^{1-\frac{1}{\delta}}}{c} \right)$.
\end{proof}

\subsection{A family of treewidth-extremal fractals for all dimensions}\label{subsec:twfractal}

We can now generalize the ideas from the previous section to obtain a family of pointsets that allow us to a get a lower bound on the treewidth of spanners for any given choice of fractal dimension.

\paragraph{Construction of treewidth-extremal fractal pointsets}.~Consider the family of pointsets defined as follows: For all integers $d>0$, we name the dimensions from $\{1,2, \ldots, d \}$ in an arbitrary manner. For all odd integers $l$ and $v$ such that $l > v$, we define each of $f^{l,v,d}(0), h^{l,v,d}_{1}(0), \ldots, h^{l,v,d}_{d}(0)$ to be a single point. We inductively define $f^{l,v,d}(i), h^{l,v,d}_{1}(i), \ldots, h^{l,v,d}_{d}(i)$ as follows: For $h^{l,v,d}_{1}(i)$, we start with a $d$-dimensional $l^{i}$ integer grid and subdivide it to get $l^{d}$ identical $l^{i-1}$ $d$-dimensional integer subgrids. Now, along every dimension $j \in \{2, \ldots, d \}$, we remove all $l^{d-1}v$ subgrids in the middle $v$ rows of the subgrids. We then replace each of the remaining $l(l-v)^{d-1}$ subgrids with copies of $h^{l,v,d}_{1}(i-1)$. The pointset obtained is $h^{l,v,d}_{1}(i)$. In general for any $m \in [d]$, we construct $h^{l,v,d}_{m}(i)$ as follows: we start with a $d$-dimensional $l^{i}$ integer grid and subdivide it to get $l^{d}$ identical $l^{i-1}$ $d$-dimensional integer subgrids. Then, along every dimension $j \neq m$, we remove all $l^{d-1}v$ subgrids in the middle $v$ rows of the subgrids. We replace each of the remaining $l(l-v)^{d-1}$ subgrids with copies of $h^{l,v,d}_{m}(i-1)$. The pointset thus obtained is $h^{l,v,d}_{m}(i)$. In order to construct $f^{l,v,d}(i)$, we start with a $d$-dimensional $l^{i}$ integer grid and subdivide it into $l^{d}$ identical $l^{i-1}$ $d$-dimensional integer subgrids. Let $S$ denote the set of the central $v^{d}$ subgrids. Note that $S$ is a $d$-dimensional grid with side length $v(l^{i-1})$. Now along every dimension $m \in [d]$, there are $\frac{(l-v)v^{d-1}}{2}$ subgrids on each side of $S$. We remove $S$ as well as all such $\frac{(l-v)v^{d-1}}{2}$ subgrids lying on either side of $S$ along every dimension. We replace each of the $\left( \frac{l-v}{2} \right)^{d}$ sub-grids in the $2^{d}$ corners with copies of $f^{l,v,d}(i-1)$. Then along every dimension $m$, we replace each of the remaining $\left(\frac{l-v}{2} \right)^{d-1} v 2^{d-1}$ subgrids with copies of $h^{l,v,d}_{m}(i-1)$. The pointset thus obtained is $f(i)$. To generate the pointset $P$ mentioned in the statement of Theorem \ref{thm:spannerhigherdim} we pick $l$ and $v$ such that $\left| \delta - \frac{\log (l(l-v)^{d-1})}{\log l} \right| \leq \eps$. Such a pair of odd numbers always exists. Let $\delta' = \frac{(\log l(l-v)^{d-1})}{\log l}$. Then, we set $P$ to be $f^{l,v,d}(k)$, where $k$ is any positive integer. We have the following lemmas regarding pointset P.

\begin{lemma} \label{lem:fractalhigherdimmodular}
For all odd integers $l$ and $v$ such that $l>v$, and for all positive integers $k>0$, we have that $\mathrm{dim_f}(f^{l,v,d}(k)) \leq \frac{(\log l(l-v)^{d-1})}{\log l}$.
\end{lemma}

\begin{proof}
Let $\eps > 0$ and $r \geq 2 \eps$. Let $j$ be the smallest non-negative integer such that $ \eps \leq (d-1)l^{j}$. We first define a family of $\eps$ coverings of $P$. For all integers $i \geq 0$, we define $f'(i), h'_{1}(i), \ldots, h'_{d}(i)$ to be $\eps$-coverings of points $f^{l,v,d}(i), h^{l,v,d}_{1}(i), \ldots, h^{l,v,d}_{d}(i)$ respectively. For $i \leq j-1$, we define $f'(i), h'{1}(i), \ldots, h'_{d}(i)$ to be a single point picked from $f^{l,v,d}(i), h^{l,v,d}_{1}(i), \ldots, h^{l,v,d}_{d}(i)$ respectively. This is because diameter of any of these sets is $\sqrt{d}l^{i}$. Thus, the choice of $\eps$ ensures that any single point forms an $\eps$-covering of its respective set. For $i \geq j$, we define $f'(i), h'_{1}(i), \ldots, h'_{d}(i)$ inductively as done in the construction of $f^{l,v,d}(i), h^{l,v,d}_{1}(i), \ldots, h^{l,v,d}_{d}(i)$. We obtain $f'(i)$ by starting with $f^{l,v,d}(i)$ and replacing copies of $f^{l,v,d}(i-1), h^{l,v,d}_{1}(i-1), \ldots, h^{l,v,d}_{d}(i-1)$ with $f'(i-1), h'_{1}(i-1), \ldots, h'_{d}(i-1)$ respectively. Similarly, we construct $h'_{1}(i), \ldots, h'_{d}(i)$. 

For a fixed $x \in \mathbf{R}^{d}$, consider the $\mathrm{ball}(x,r)$. Let $N$ be a $2 \eps$-net of $P$ and $C$ be an $\eps$-covering of $P$. We have $\vert N \cap \mathrm{ball}(x,r) \vert \leq |C \cap \mathrm{ball}(x, r + \eps)|$. This is because every point in $N$ is covered by a unique point of $C$, since $N$ is a $2 \eps$-net and $C$ is an $\eps$-covering of $P$. We may set $C$ to be $f'(k)$, since $f'(k)$ is an $\eps$-covering of $P$. Let $t \geq 0$ be the smallest integer such that $l^{t} \geq 2r$. We have $\vert f'(k) \cap \mathrm{ball}(x, r + \eps) \vert \leq (2d+1) \vert f'(t) \vert $ since the diameter of $\mathrm{ball}(x, r+ \eps)$ is at most $3r$. Thus, we have $\vert N \cap \mathrm{ball}(x,r) \vert \leq (2d+1) \vert f'(t) \vert \leq (2d+1) \vert f'(j-1) \vert d \cdot l^{t-j+1} (l-v)^{(d-1)(t-j+1)} $. The second inequality holds because $\vert f'(k) \vert \leq dl^{k} (l-v)^{k(d-1)}$. Thus,
\begin{align*}
\vert N \cap ball(x,r) \vert &\leq (2d+1)d \vert f'(j-1) \vert \cdot l^{(t-j+1) (\log_{l} l(l-v)^{d-1})} \\
&\leq (2d+1)d \left(\frac{(d-1)l^{j}}{\eps} \right)^{\log_{l} l(l-v)^{d-1}} \cdot l^{(t-j+1) (\log_{l} l(l-v)^{d-1})} \\
&= (2d+1)d \left( \frac{(d-1)l^{t+1}}{\eps} \right)^{\log_{l} l(l-v)^{d-1}} \\
&< (2d+1)d \left( \frac{(d-1)2rl^{2}}{\eps} \right)^{\log_{l} l(l-v)^{d-1}} \\
&= (2d+1)d(4(d-1)l^{2})^{\log_{l} l(l-v)^{d-1}} \cdot \left( \frac{r}{2 \eps} \right)^{\log_{l} l(l-v)^{d-1}}. 
\end{align*}
Now, from the definition of fractal dimension and the fact that we have $ \vert N \cap ball(x,r) \vert = O\left( \left( \frac{r}{2 \eps} \right)^{\frac{\log l(l-v)^{d-1}}{\log l}} \right)$, it follows that $\mathrm{dim}_{f}(f^{l,v,d}(k)) \leq \frac{\log l(l-v)^{d-1}}{\log l}$.
\end{proof}

\begin{lemma} \label{lem:fractalhigherdim}
$\mathrm{dim}_{f}(P) \leq \delta'$.
\end{lemma}
\begin{proof}
This follows from Lemma \ref{lem:fractalhigherdimmodular} when applied to our choice of $l$,$v$ and $P$.
\end{proof}

\begin{lemma} \label{lem:twhigherdim}
Let $G$ be a $c$-spanner of $P$, where $c \geq 1$. Then, $\tw(G) = \Omega \left(  \frac{(l-v)^{k(d-1)}}{c^{d-1}} \right)$.
\end{lemma}

\begin{proof}
We calculate the number of full rows along an arbitrary dimension $i$ since the number of full rows along every dimension is the same. We recall that a full row in $G$ is such that consecutive points in the row are at distance $1$ from each other. We observe that for any $i \in [d]$, the number of full rows in $h^{l,v,d}_{i}(k)$ is equal to $(l-v)^{k(d-1)}$. This is because the total number of points in $h^{l,v,d}_{i}(k)$ is $l^{k} (l-v)^{k(d-1)}$ and every point of $h^{l,v,d}_{i}(k)$ belongs to a unique full row. Note that number of points in a full row of $h^{l,v,d}_{i}(k)$ is $l^{k}$. We observe that the number of full rows in $f^{l,v,d}(k)$ is equal to the number of full rows in $h^{l,v,d}_{i}(k)$. Therefore, the number of full rows in $f^{l,v,d}(k)$ along any dimension $i$ is equal to $(l-v)^{k(d-1)}$.

Next we use a similar argument to that of theorem \ref{thm:ddimensionalgridspanner}. Let us denote by $S'_i$ the set of full rows of $f^{l,v,d}(k)$ parallel to the i\textsuperscript{th} axis. Further for all $j \neq i$ we denote by $r^i(j,m)$ the $m$\textsuperscript{th} set of rows in $S'_i$ along the $j$\textsuperscript{th} axis. We next pick for all $i \in [d]$ a subset of rows in $S'_i$ $S_i = \{ r^i(j,m) : j \neq i \text{ and } ( m \mod (c+1)) = 1 \}$. Let $X = S_1 \cap S_2 \cap \ldots \cap S_d$. Consider for all $i \in [d]$ the points of any pair of full rows $R,T \in S'_i$. Similar to the argument used in theorem \ref{thm:ddimensionalgridspanner} we have that for any pair of consecutive points $x_{1}, x_{2}$ in $R$ and for any pair of consecutive points $y_{1}, y_{2}$ in $T$, no shortest path in $G$ joining $x_{1}$ and $x_{2}$ can intersect any shortest path in $G$ joining $y_{1}$ and $y_{2}$. Thus like in the earlier proof for all $1 \leq i \leq d$ we can obtain a set of vertex disjoint paths $Q_i$ in $G$ that traverse the points in $S'_i$. Finally just as in theorem \ref{thm:ddimensionalgridspanner} we again consider $H$ a subgraph of $G$ as follows. $H$ consists of the points in $X$ and for any pair of adjacent points $p,q \in X$ such that $p$ and $q$ differ only along say the $i$\textsuperscript{th} coordinate $H$ also consists of the sub-path of the corresponding path in $Q_i$ connecting $p$ and $q$. Now contracting these paths in $H$ to get edges between adjacent points in $X$ results in a $d$-dimensional grid with side $ \frac{(l-v)^{k}}{c} $. Thus, we conclude that $\tw(G) = \Omega \left( \left( \frac{(l-v)^{kd}}{c^{d}} \right)^{\frac{d-1}{d}} \right) = \Omega \left( \frac{(l-v)^{k(d-1)}}{c^{d-1}} \right)$.  
\end{proof}
 
\begin{proof}[Proof of Theorem \ref{thm:spannerhigherdim}]
Using Lemma \ref{lem:twhigherdim} and Lemma \ref{lem:fractalhigherdim}, we get that 
$$\tw(G) = \Omega \left( \frac{(l-v)^{k(d-1)}}{c^{d-1}} \right) = \Omega \left( \frac{l^{k \delta' \left(1- \frac{1}{\delta'}\right)}}{c^{d-1}} \right) = \Omega \left( \frac{n^{1- \frac{1}{\delta'}}}{c^{d-1}} \right) .$$ 
This combined with the result of Lemma \ref{lem:fractalhigherdim} proves the statement of the theorem.
\end{proof}

\section{Running time lower bound for Independent Set of Unit Balls} \label{sec:independent-balls}
In this section, we present the proof of Theorem \ref{thm:independent_balls}. Our argument uses a reduction from a type of \csp{} called the Geometric \csp{}. The definitions in this section are taken from \cite{MS14}.
\begin{definition}[\textbf{The \csp}] \cite{MS14}
The input instance $I$ of a constraint satisfaction problem is a triple $(V,D,C)$, where $V$ is a set of variables that can take values in the domain $D$, and $C$ is a set of constraints, with each constraint being a pair $ \langle s_i,R_i \rangle$ such that:
\begin{itemize}
\item $s_i$ is a tuple of variables of size $m_i$.
\item $R_i$ is an $m_i$-ary relation over $D$.
\end{itemize}
A valid solution to the problem is an assignment of values from $D$ to each of the variables in $V$ such that for all constraints $ \langle s_i,R_i \rangle$, the assignment for each tuple $s_i$ is in $R_i$.
\end{definition}

For our purposes, we only need to consider the case where the constraints are binary, or in other words, for all $i$, we have that $m_i = 2$. We may assume that the input size $|I|$ of a binary CSP instance is a polynomial in $|V|$ and $|D|$. 
 
For an instance $I$ of the \csp{}, the \emph{primal graph} is a graph $G$ with vertex set $V$ and an edge between $u,w \in V$ if and only if there exists a constraint $ \langle s_i,R_i \rangle \in C$, such that $s_i = (u,w)$. 

Let $R[n,d]$ denote the $d$-dimensional grid with vertex set $[n]^d$ and let ${\cal{R}}_{d}$ denote the set of graphs $R[n,d]$ for all $n \geq 1$. 

\begin{definition}[\textbf{The $\leq$-CSP}] \cite{MS14}
A $d$-dimensional geometric $\leq$-CSP is a constraint satisfaction problem of the following form: The set of variables $V$ is a subset of vertices of $R[n,d]$ for some $n$ and the primal graph is an induced subgraph of $R[n,d]$. The domain is $[\Delta]^{d}$ for some integer $\Delta \geq 1$. The instance can contain arbitrary unary constraints but the binary constraints are of a special form. A geometric constraint is a constraint $\left\langle \left( \mathbf{a}, \mathbf{a'} \right),R \right\rangle$ with $\mathbf{a'} = \mathbf{a} + \mathbf{e}_{i}$ such that
$$ R = \{((x_{1}, \ldots, x_{d}), (y_{1}, \ldots, y_{d})) ~|~ x_{i} \leq y_{i} \} $$
This means that, if variables $\mathbf{a}$ and $\mathbf{a'}$ are adjacent with $\mathbf{a'}$ being larger by one in the $i$-th coordinate, then the $i$-th coordinate of the value of $\mathbf{a}$ is at most as large as the $i$-th coordinate of the value of $\mathbf{a'}$.
\end{definition}

We will use the following theorem from \cite{MS14} in the proof of Theorem \ref{thm:independent_balls}.
We remark that the condition $|V|=\Theta(n^d)$ is implicit in \cite{MS14}.

\begin{theorem} \label{thm:geometric-csp} \textnormal{[\cite{MS14}, Theorem 2.20]}
If for some fixed $d \geq 1$, there is an $f(\vert V \vert)n^{o(\vert V \vert^{1-1/d})}$ time algorithm for $d$-dimensional geometric $\leq$-CSP for some function $f$, where $|V|=\Theta(n^d)$, then ETH fails.
\end{theorem}

\begin{figure}
\begin{center}
\scalebox{0.15}{\includegraphics{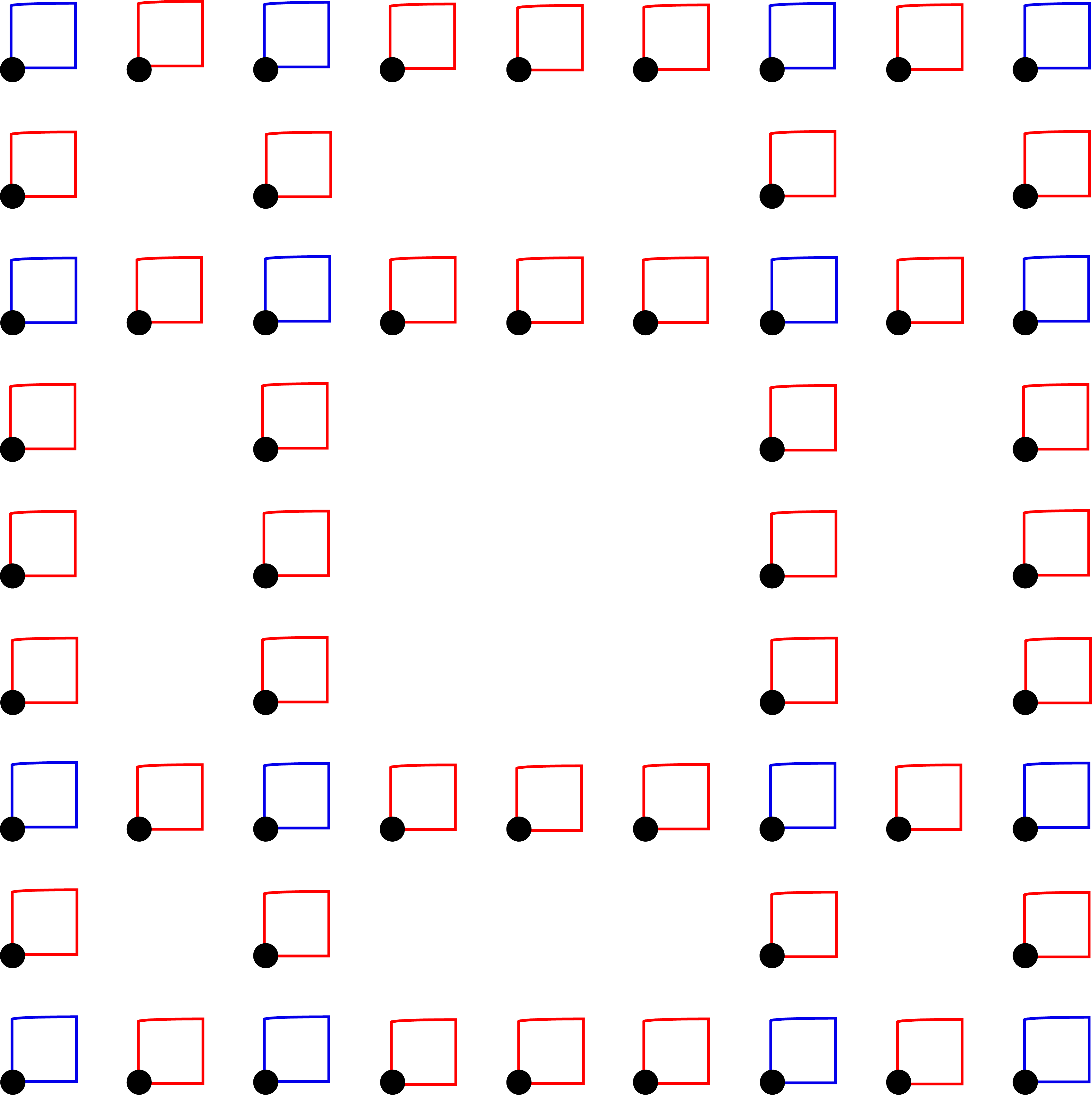}}
\caption{The case $d=2,~ l=3,~ v=1$ and $m=2$. Boxes contain centers of balls in $B$. The blue boxes are obtained from variables in $I$ and the red boxes are obtained from variables in $I' \setminus I$. }
\label{fig:independent_disks}
\end{center}
\end{figure}

\paragraph{Construction of $\leq$-CSP $I'$}.~Given $\delta \in (1,d)$ and $\eps' > 0$, we can find odd integers $l$ and $v$ such that $\left\vert \delta - \frac{\log l(l-v)^{d-1}}{\log l} \right\vert \leq \eps'$. Let $\delta' = \frac{\log l(l-v)^{d-1}}{\log l}$. Let $I$ be a $d$-dimensional $\leq$-CSP instance with variables $V$ and domain $[\Delta]^{d}$, where $\Delta$ is any positive integer. Let the primal graph of $I$ be $R[n,d]$. We now define a new $d$-dimensional $\leq$-CSP $I'$, with variables $V'$ and domain $[\Delta]^d$, such that $|V'| = O \left(\vert V \vert^{\frac{d-1}{d} \cdot \frac{\delta'}{\delta' - 1}} \right)$, and the primal graph of $I'$ is $R[n_{new},d]$, where $n_{new} = n^{O(1)}$. Let $m > 0$ be the smallest integer such that $ l^{\frac{(m-1)(\delta' - 1)}{d-1}} < n \leq l^{\frac{m(\delta' - 1)}{d-1}}$.  We construct $f^{l,v,d}(m)$, as in the proof of Theorem \ref{thm:spannerhigherdim}. From Lemma \ref{lem:twhigherdim}, we know that $f^{l,v,d}(m)$ contains a subset of points that form a $d$-dimensional grid of side length $l^{\frac{m(\delta' - 1)}{d-1}}$. Let $M$ denote this grid contained in $f^{l,v,d}(m)$. Since the variables $V$ lie on the grid $R[n,d]$ and $n^d \leq l^{\frac{md(\delta' - 1)}{d-1}}$, we can place the variables $V$ on grid $M$ such that their position relative to each other in $M$ is the same as in $R[n,d]$. By abuse of notation, we refer to the subset of $M$ containing variables $V$ as $V$. We refer to the pointset $f^{l,v,d}(m) \setminus M$ as $C$. We observe that the points in $C$ connect adjacent points of $M$. We now define a new set of variables $V' = V \cup C$. If $\mathbf{a} \in V$, then the set of unary constraints for $\mathbf{a}$ is $R_{\mathbf{a}}$. If $\mathbf{a} \in C$, then $\mathbf{a}$ belongs to a chain of points that connects two points of $M$ along some dimension $i$, where $i \in [d]$. For an $\mathbf{a} \in C$ connecting two points of $M$ along dimension $i$, we define $R_{\mathbf{a}}$ as follows: 
$$R_{\mathbf{a}} = \{(a_1, \ldots, a_{i-1}, a_i, a_{i+1} \ldots,a_d) ~|~ a_j = 0 ~ \forall ~ j \neq i ~\mathrm{and}~ a_i \in [\Delta]\}.$$
Thus for every $\mathbf{a} \in C$, we have $|R_{\mathbf{a}}| = \Delta$.
We define $R_{\mathbf{a}}$ for all variables $\mathbf{a} \in C$ in a similar manner. The binary constraints on variables $V'$ are described as follows: a binary constraint is a constraint $\langle(\mathbf{a}, \mathbf{a'}),R \rangle$, with $\mathbf{a'} = \mathbf{a} + \mathbf{e}_i$ in $f^{l,v,d}(m)$ such that
$$ R = \{((x_1, \ldots, x_d),(y_1, \ldots, y_d))~|~x_i \leq y_i \}. $$
Let $I'$ denote the new $\leq$-geometric CSP with variables $V'$, and unary and binary constraints as defined above. From the choice of $m$, we have that $|V'| \leq l^{\delta'} \cdot n^{(d-1) \cdot \frac{\delta'}{\delta' - 1}}$. Thus $|V'| = O\left(|V|^{\frac{d-1}{d} \cdot \frac{\delta'}{\delta' - 1}}\right)$,
where we use the fact that $\delta$, $\delta'$, and $l$ are fixed constants.
Similarly, we get that $ n_{new} = n^{O(1)}$.

\begin{lemma} \label{lem:newcsp}
$I'$ is satisfiable if and only if $I$ is satisfiable. 
\end{lemma}

\begin{proof}
Let $h$ be a satisfying assignment for $I'$. This means that for every $\mathbf{a} \in V'$, $h(\mathbf{a}) \in R_{\mathbf{a}}$ and if $\mathbf{a}, \mathbf{a'} \in V'$ are such that $\mathbf{a'} = \mathbf{a} + \mathbf{e}_{i}$ in $f^{l,v,d}(m)$, then the $i$-th coordinate of $h(\mathbf{a})$ is at most the $i$-th coordinate of $h(\mathbf{a'})$. Let $\mathbf{a}, \mathbf{a'} \in V$ be such that $\mathbf{a'} = \mathbf{a} + \mathbf{e}_{i}$ in $R[n,d]$. Now $\mathbf{a}$ and $\mathbf{a'}$ are not adjacent in $f^{l,v,d}(m)$, but there is a chain of variables $\mathbf{a^{1}}, \mathbf{a^2}, \ldots \in C$ connecting $\mathbf{a}$ and $\mathbf{a'}$ along dimension $i$. Since $\mathbf{a^{1}} = \mathbf{a} + e_{i}$ in $f^{l,v,d}(m)$, the $i$-th coordinate of $h(\mathbf{a})$ is at most the $i$-th coordinate of $h(\mathbf{a^1})$. Similarly, we have $\mathbf{a^{2}} = \mathbf{a^{1}} + e_{i}$ in $f^{l,v,d}(m)$ and therefore the $i$-th coordinate of $h(\mathbf{a^1})$ is at most the $i$-th coordinate of $h(\mathbf{a^2})$. For every pair of consecutive variables $\mathbf{a^{j}}, \mathbf{a^{j+1}}$ in this chain, we have $\mathbf{a^{j+1}} = \mathbf{a^{j}} + e_i$ in $f^{l,v,d}(m)$ and therefore the $i$-th coordinate of $h(\mathbf{a^{j}})$ is at most the $i$-th coordinate of $h(\mathbf{a^{j+1}})$. Since $\leq$ is a transitive relation, we get that the $i$-th coordinate of $h(\mathbf{a})$ is at most the $i$-th coordinate of $h(\mathbf{a'})$. Therefore $h$ is a satisfying assignment for $I$. 

Let $h$ be a satisfying assignment for $I$. This means that for every $\mathbf{a} \in V$, $h(\mathbf{a}) \in R_{\mathbf{a}}$ and if $\mathbf{a}, \mathbf{a'} \in V$ are such that $\mathbf{a'} = \mathbf{a} + \mathbf{e}_{i}$ in $R[n,d]$, then the $i$-th coordinate of $h(\mathbf{a})$ is at most the $i$-th coordinate of $h(\mathbf{a'})$. Let $\mathbf{a^{1}}, \mathbf{a^2}, \ldots \in C$ be a chain of variables connecting $\mathbf{a}$ and $\mathbf{a'}$ in $f^{l,v,d}(m)$ along dimension $i$. For every variable $\mathbf{a^j}$ in this chain, define $h(\mathbf{a^j}) = (0, \ldots, 0, x_i,0,\ldots,0)$, where $x_{i}$ is equal to the $i$-th coordinate of $h(\mathbf{a})$. It can be checked that $h(\mathbf{a^j}) \in R(\mathbf{a^j})$. Since every variable in $V' \setminus V$ connects some pair of adjacent variables of $I$, we define $h$ on $V' \setminus V$ as above. It is straightforward to see that $h$ is a satisfying assignment for $I'$. 
\end{proof}

\begin{proof}[Proof of Theorem \ref{thm:independent_balls}]
We describe the construction of the set of balls $B$. We work with open balls of diameter $1$. Then, we have that two balls are non-intersecting if and only if the distance between their centers is at least $1$. Let $\alpha = \frac{1}{d \Delta^{2}}$. Let $\mathbf{a} = (a_{1}, \ldots, a_{d})$ be a variable in $V'$ and let $ \langle (\mathbf{a}), R_{\mathbf{a}} \rangle $ be the unary constraint on $\mathbf{a}$. For every constraint $\mathbf{x} = (x_{1}, \ldots, x_{d}) \in R_{\mathbf{a}}$, we consider the point $\mathbf{a} + \alpha \mathbf{x} = (a_{1} + \alpha x_{1}, \ldots, a_{d} + \alpha x_{d})$ and add an open ball of radius $\frac{1}{2}$ centered at $\mathbf{a} + \alpha \mathbf{x}$ to $B$. We do this for every variable $\mathbf{a} \in V'$. We have $\dimF(f^{l,v,d}(m)) = \delta'$, for every $\mathbf{a} \in V'$,  $|R_{\mathbf{a}}| \leq \Delta^{d}$ and for every $\mathbf{x} \in R_{\mathbf{a}}$, $d(\mathbf{a}, \mathbf{a} + \alpha \mathbf{x}) \leq \frac{1}{\sqrt{2}}$. Thus using Lemma \ref{lem:substitution}, we get that the fractal dimension of the set of centers of the balls in $B$ is $\delta'$. We refer the reader to Figure \ref{fig:independent_disks} for a visualization of the set of centers.

Let $B_{\mathbf{a}}$ denote the set of balls obtained from the unary constraints of $\mathbf{a}$. Note that all the balls in $B_{\mathbf{a}}$ intersect each other. Therefore if we find a collection $B' \subseteq B$ of pairwise non-intersecting balls, then $|B'| \leq |V'|$ and $|B'| = |V'|$ is possible if $B'$ contains exactly one ball from every $B_{\mathbf{a}}$. 

Let $B_{1}$ be the ball centered at $\mathbf{a} + \alpha \mathbf{x}$ and $B_{2}$ be the ball centered at $\mathbf{a} + \mathbf{e}_{i} + \alpha \mathbf{x'}$ for some $\mathbf{a} \in V'$, $i \in [d]$, $\mathbf{x} = (x_1, \ldots, x_d) \in [\Delta]^d$ and $\mathbf{x'} = (x'_1, \ldots, x'_d) \in [\Delta]^d$. We claim that $B_{1}$ and $B_{2}$ are non-intersecting if and only if $x_{i} \leq x'_{i}$. We have that $B_{1}$ and $B_{2}$ intersect if and only if the distance between their centers is less than $1$. The distance between their centers is given by
$$ \sum_{j=1}^{i-1} \alpha^2 (x'_{j} - x_{j})^{2} + (1+ \alpha(x'_i - x_i))^{2} + \sum_{j=i+1}^{d} \alpha^2 (x'_{j} - x_{j})^2 \leq (d-1) \alpha^2 \Delta^2 + (1+ \alpha(x'_i - x_i))^{2}. $$
The above inequality is because for every $j \in [d]$, we have $(x'_j - x_j)^2 \leq \Delta^2$. If $x_i > x'_i$, then we get that $(1+ \alpha(x'_i - x_i))^{2} < (1- \alpha)^2$. This gives 
$$ (d-1) \alpha^2 \Delta^2 + (1+ \alpha(x'_i - x_i))^{2} < d \alpha^2 \Delta^2 + (1- \alpha)^2 \leq \alpha + (1- \alpha)^2 = 1 - \alpha + \alpha^2 < 1. $$
The first inequality holds because $\alpha = \frac{1}{d \Delta^2}$. Thus we have shown that if $B_{1}$ and $B_{2}$ are non-intersecting, then $x_{i} \leq x'_{i}$. On the other hand, if $x_{i} \leq x'_{i}$ then $(1+ \alpha(x'_i - x_i))^{2} \geq 1$. Since $B_{1}$ and $B_{2}$ are open balls, they are non-intersecting.

If $\mathbf{a}$ and $\mathbf{a'}$ are not adjacent in $f^{l,v,d}(m)$, then for any $\mathbf{x}, \mathbf{x'} \in [\Delta]^d$ the balls centered at $\mathbf{a} + \alpha \mathbf{x}$ and $\mathbf{a'} + \alpha \mathbf{x'}$ cannot intersect because the square of the distance between the centers is at least $2(1 - \alpha \Delta)^2 > 1$.

Let $g$ be a satisfying assignment for $I'$. For every variable $\mathbf{a} \in V'$, we select the ball $\mathbf{a} + \alpha g(\mathbf{a}) \in B_{\mathbf{a}} $. If $\mathbf{a}$ and $\mathbf{a'}$ are not adjacent in $f^{l,v,d}(m)$, then the corresponding balls $\mathbf{a} + \alpha g(\mathbf{a})$ and $\mathbf{a'} + \alpha g(\mathbf{a'})$ do not intersect. If $\mathbf{a}$ and $\mathbf{a'}$ are adjacent in $f^{l,v,d}(m)$, then there is a geometric binary constraint on $\mathbf{a}$ and $\mathbf{a'}$. Therefore, if say $\mathbf{a'} = \mathbf{a} + \mathbf{e}_{i}$, then the binary constraint ensures that the $i$-th coordinate of $g(\mathbf{a})$ is at most the $i$-th coordinate of $g(\mathbf{a'})$. Thus, the balls centered at $\mathbf{a} + \alpha g(\mathbf{a})$ and $\mathbf{a'} + \alpha g(\mathbf{a'})$ do not intersect. Therefore the balls with centers $\mathbf{a} + \alpha g(\mathbf{a}),~ \mathbf{a} \in V'$ form a set of $|V'|$ pairwise non-intersecting balls. 

Conversely, let $B_{0} \subseteq B$ be a set of $\vert V' \vert$ pairwise non-intersecting balls. This is possible only if for every $\mathbf{a} \in V'$, the set $B_{0}$ contains a single ball from $B_{\mathbf{a}}$. Suppose the unique ball in $B_{\mathbf{a}} \bigcap B_{0}$ is centered at $\mathbf{a} + \alpha g(\mathbf{a})$, for some $g(\mathbf{a}) \in [\Delta]^{d}$. We claim that $g$ is a satisfying assignment for $I'$. It satisfies the unary constraints because $\mathbf{a} + \alpha g(\mathbf{a}) \in B_{\mathbf{a}}$ implies that $g(\mathbf{a}) \in R_{\mathbf{a}}$. Now let $\mathbf{a}$ and $\mathbf{a'} = \mathbf{a} + \mathbf{e}_{i}$ be adjacent variables in $f^{l,v,d}(m)$. Since the balls centered at $\mathbf{a} + \alpha g(\mathbf{a})$ and $\mathbf{a'} + \alpha g(\mathbf{a'})$ do not intersect, we get that the $i$-th coordinate of $g(\mathbf{a})$ is at most the $i$-th coordinate of $g(\mathbf{a'})$. Thus the geometric binary constraint on $\mathbf{a}$ and $\mathbf{a'}$ is satisfied. 

\end{proof}

\begin{Remark}
We can similarly show that assuming Exponential Time Hypothesis, for any $\delta \in (1,d)$ and any $\eps > 0$, the problem of finding $k$-pairwise non-intersecting $d$-dimensional axis parallel unit cubes in a collection of $n$ cubes with centers having fractal dimension at most $\delta$ cannot be solved in time $f(k) n^{O(k^{1-1/(\delta - \eps)})}$, for any computable function $f$. Given a $\leq$-CSP instance $I$, we reduce $I$ to another $\leq$-CSP instance $I'$ as explained in the beginning of this section. We then use the construction and analysis as in the proof of Theorem $3.2$ of \cite{MS14}, replacing $I$ with $I'$. Using Theorem \ref{thm:geometric-csp} and the proof of Theorem \ref{thm:independent_balls}, we get the desired result. 
\end{Remark}

\section{Running time lower bounds for TSP} \label{sec:TSP}
In this section we provide a proof of Theorem \ref{thm:tsplowerboundmain}. We consider the case of $d = 2$ in subsection \ref{subsec:lowerboundtsp} and $d>2$ in \ref{subsec:lowerboundtsparbitdim}. Our analysis uses reductions of the \ecp{} and the \csp{} for the two cases respectively.

\subsection{Lower bound for TSP in $\mathbb{R}^2$}\label{subsec:lowerboundtsp}
We show a running time lower bound for TSP, on pointsets of arbitrary fractal dimension in $\mathbf{R}^2$.
\begin{theorem}\label{thm:tsplowerbound}
For all $\delta\in (1,2)$ and for all $\delta'<\delta$, if Euclidean TSP in $\mathbb{R}^2$ on all pointsets of size $n$ and fractal dimension at most $\delta$ can be solved in time $2^{O(n^{1-1/\delta'})}$, then ETH fails.
\end{theorem}

We use an argument inspired by the NP-hardness proof of Euclidean TSP in the plane due to Papadimitriou \cite{papadimitriou1977euclidean}. The proof involves a reduction of an instance of the \ecp{} to a TSP instance. We reuse the following gadgets and definitions introduced by Papadimitriou in our construction. Each gadget (called configuration in \cite{papadimitriou1977euclidean}) is a specific arrangement of points.

\paragraph{1-chain}.~It is a configuration consisting of a set of points $p_1, \ldots , p_k \in \mathbb{R}^2$ where for all $i$, $ \lVert p_{i+1} - p_i \rVert = 1$. Furthermore, $p_i$ and $p_{i+1}$ differ in a single coordinate. This configuration is depicted in Figure \ref{fig:1-chain}-a (along with a schematic abbreviation in Figure \ref{fig:1-chain}-b).

\begin{figure}[H]
\begin{center}
\scalebox{0.10}{\includegraphics{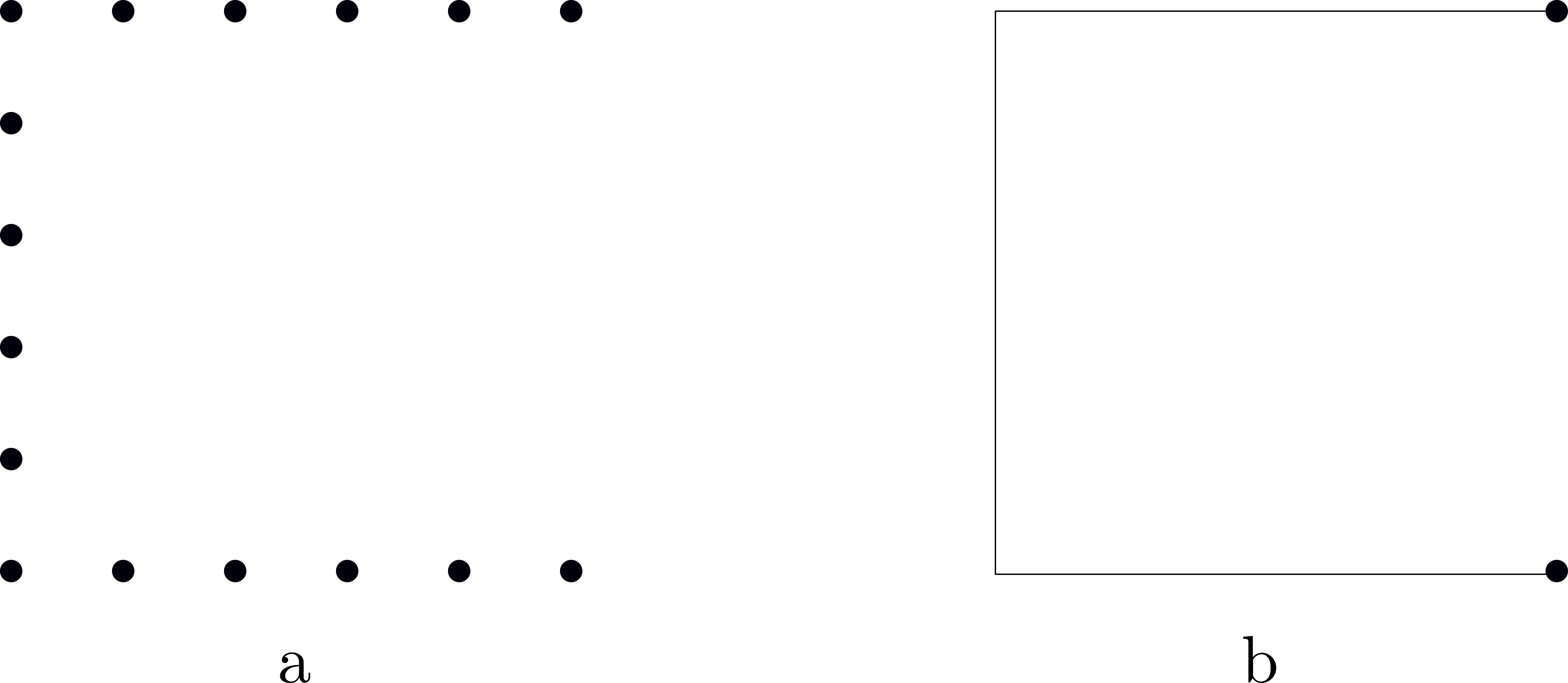}}
\caption{1-chain}
\label{fig:1-chain}
\end{center}
\end{figure}

\paragraph{2-chain}.~This configuration consists of two parallel rows of points as depicted in Figure \ref{fig:2-chain}-a. The distance between adjacent points in each row is 2 and the distance between the two rows is 1. A schematic abbreviation of a 2-chain is depicted in Figure \ref{fig:2-chain}-b. This configuration can be traversed optimally in two modes referred to as mode 2 and mode 1 which are depicted in Figure \ref{fig:2-chain}-c and \ref{fig:2-chain}-d.

\begin{figure}[H]
\begin{center}
\scalebox{0.10}{\includegraphics{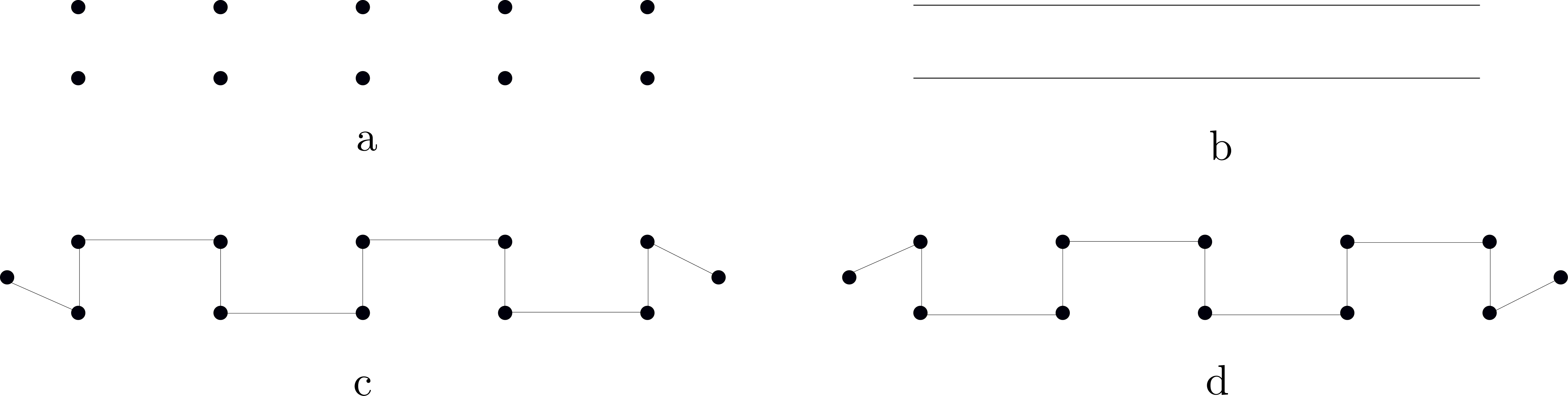}}
\caption{2-chain}
\label{fig:2-chain}
\end{center}
\end{figure}

\paragraph{configuration - H}.~This configuration is depicted in Figure \ref{fig:config_H}-a along with an abbreviation in Figure \ref{fig:config_H}-b. The width of the configuration is $8$ and the height is $7$. It can be traversed optimally in 4 ways two of which are depicted in Figure \ref{fig:config_H}-c and Figure \ref{fig:config_H}-d (the other two are variations starting and ending on points from the top). 

\begin{figure}[H]
\begin{center}
\scalebox{0.10}{\includegraphics{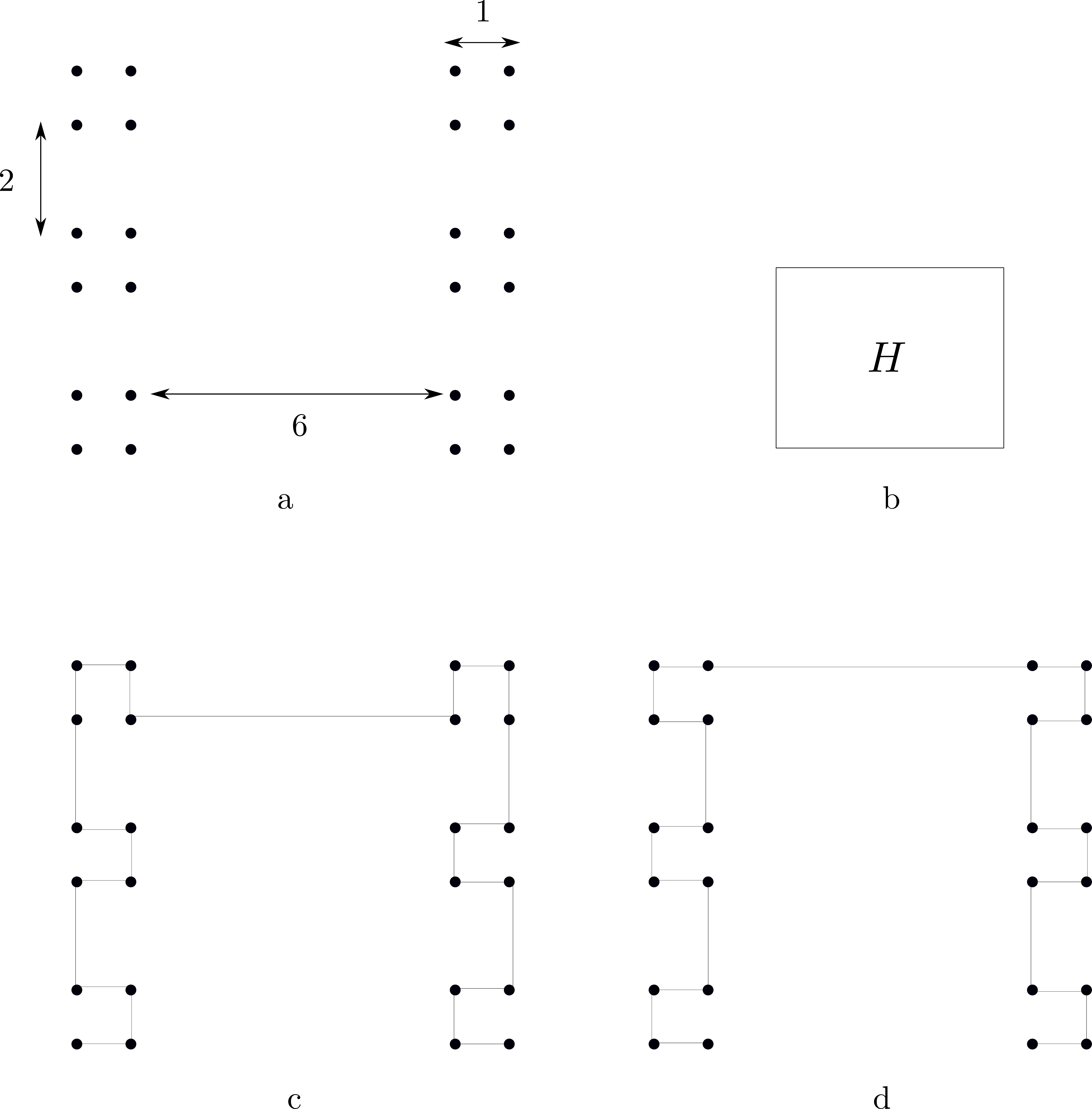}}
\caption{configuration H}
\label{fig:config_H}
\end{center}
\end{figure}

\paragraph{configuration - A}.~This configuration is depicted in Figure \ref{fig:config_A}-a (abbreviation in Figure \ref{fig:config_A}-b). It has height $4$ and width $8$. Two of the modes in which it can be traversed are depicted in 
Figure \ref{fig:config_A}-c and Figure \ref{fig:config_A}-d.
\begin{figure}[H]
\begin{center}
\scalebox{.10}{\includegraphics{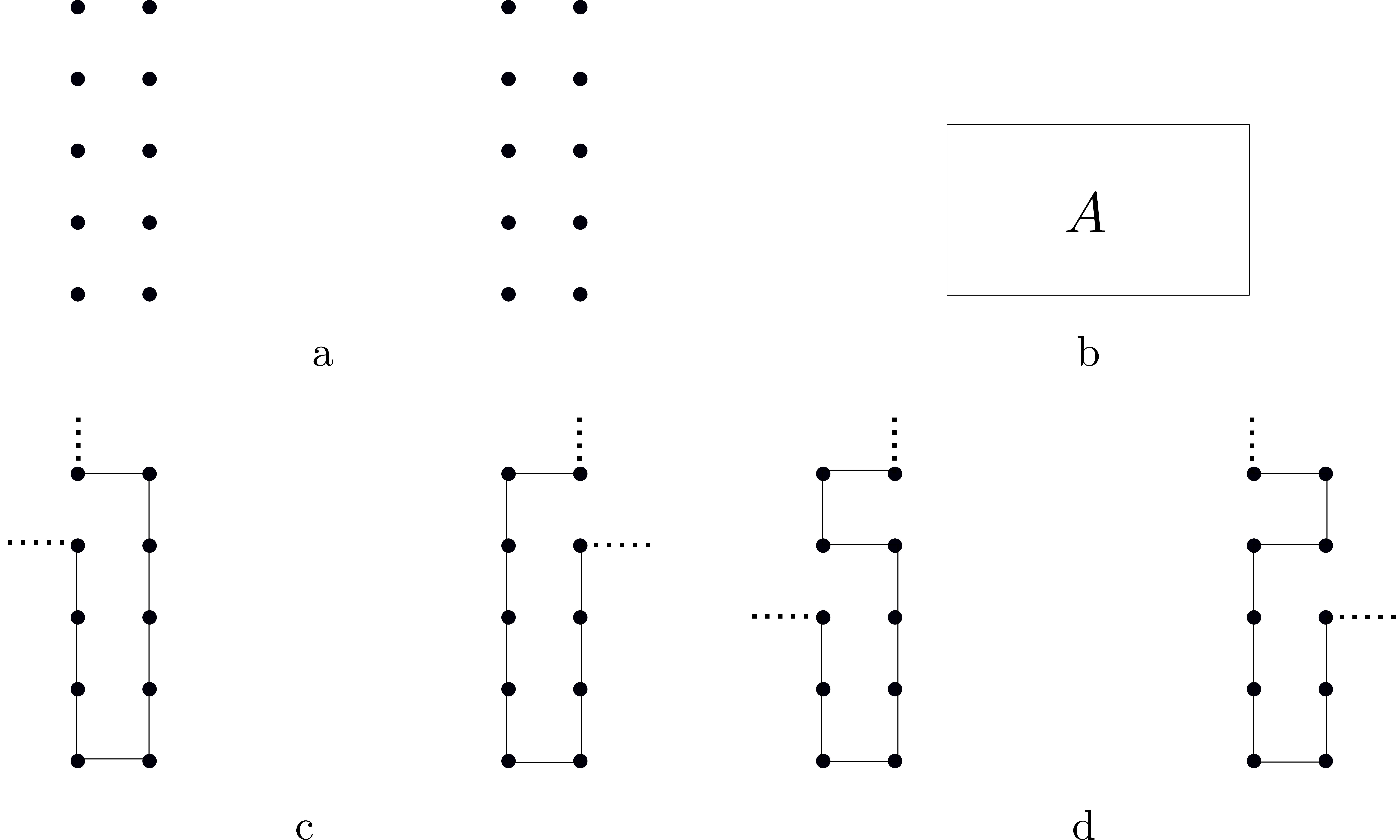}}
\caption{configuration A}
\label{fig:config_A}
\end{center}
\end{figure}

\paragraph{configuration - B}.~This configuration is depicted in Figure \ref{fig:config_B}-a (abbreviation in Figure \ref{fig:config_B}-b). It has height $4$ and width $8$. If it is entered via a 2-chain traversed in mode-2, then it can only be traversed optimally if the tour traverses it as depicted in 
Figure \ref{fig:config_B}-c. On the other hand if it is entered via a 2-chain traversed in mode-1, then it can visit points above or below it while traversing  optimally as depicted in 
Figure \ref{fig:config_B}-d.

\paragraph{$b$-component}.~Let $P$ be a set of points supporting a distance function $d$. We call $S \subset P$ a $b$-component if $S$ is maximal with respect to the following properties: 
\begin{description}
\item{1.} For all $u \in S$, $\min \{ d(u,v)| v \not\in S \} \geq b.$ 
\item{2.} For all $u \in S$, $\min \{ d(u,v)| v \in S \} < b.$ 
\end{description}

\paragraph{$k$-path}.~A $k$-path of $P$ is a set of $k$ vertex disjoint, not-closed paths that cover $P$.

\paragraph{$b$-compact}.~$S \subset P$ is $b$-compact if for all positive integers $k$, there exists a $k$-path of $S$ that has length less than $b$ plus the length of an optimal $k + 1$-path.

\begin{figure}[H]
\begin{center}
\scalebox{0.10}{\includegraphics{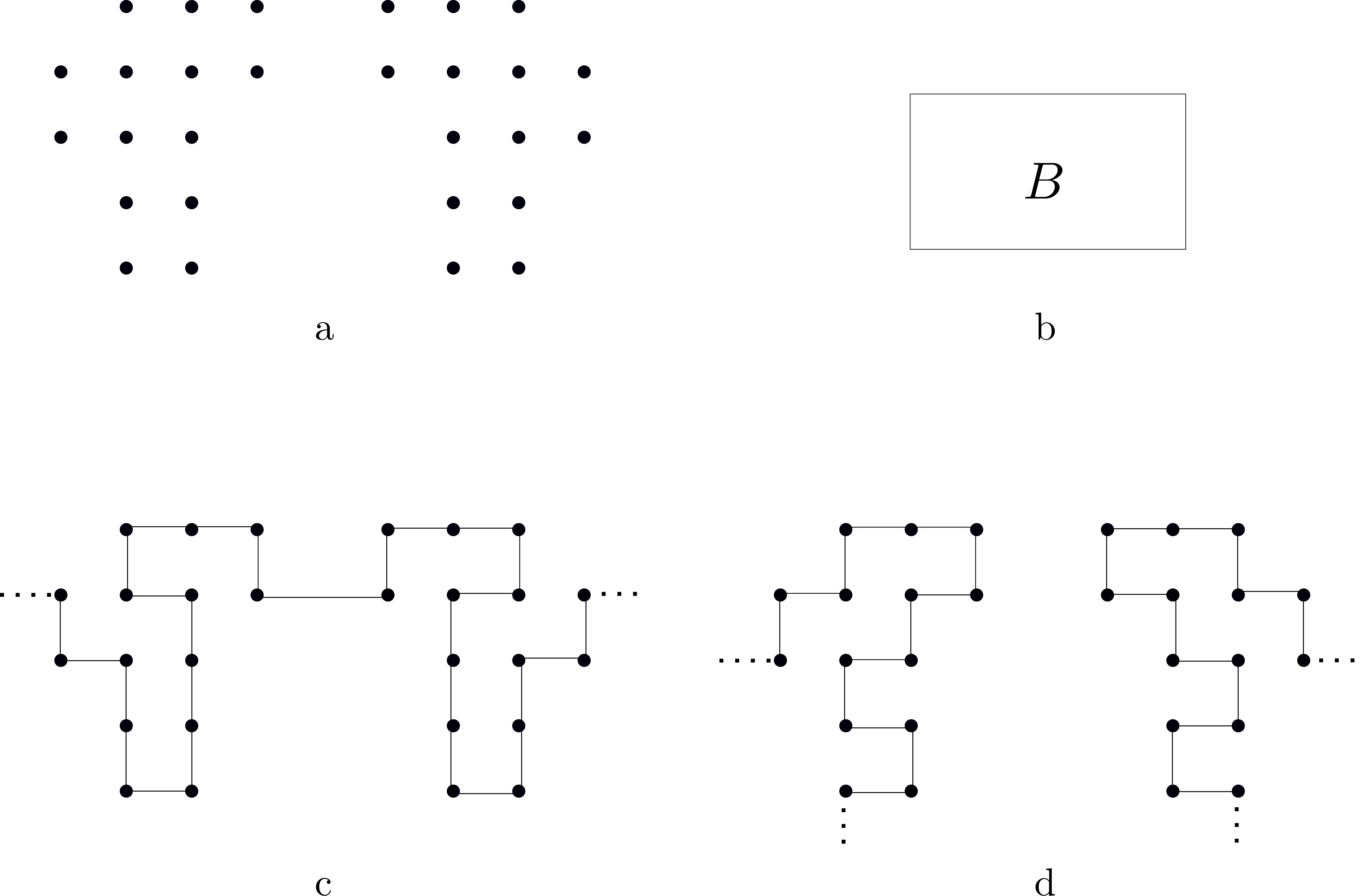}}
\caption{configuration B}
\label{fig:config_B}
\end{center}
\end{figure}

\paragraph{vertical-chain}.~This configuration consists of two parallel columns of points. The distance between adjacent points in each column is 1 and the distance between the two columns is 1. This configuration is depicted in Figure \ref{fig:vertical-chain}-a (abbreviation in Figure \ref{fig:vertical-chain}-b).

\begin{figure}[h]
\begin{center}
\scalebox{0.10}{\includegraphics{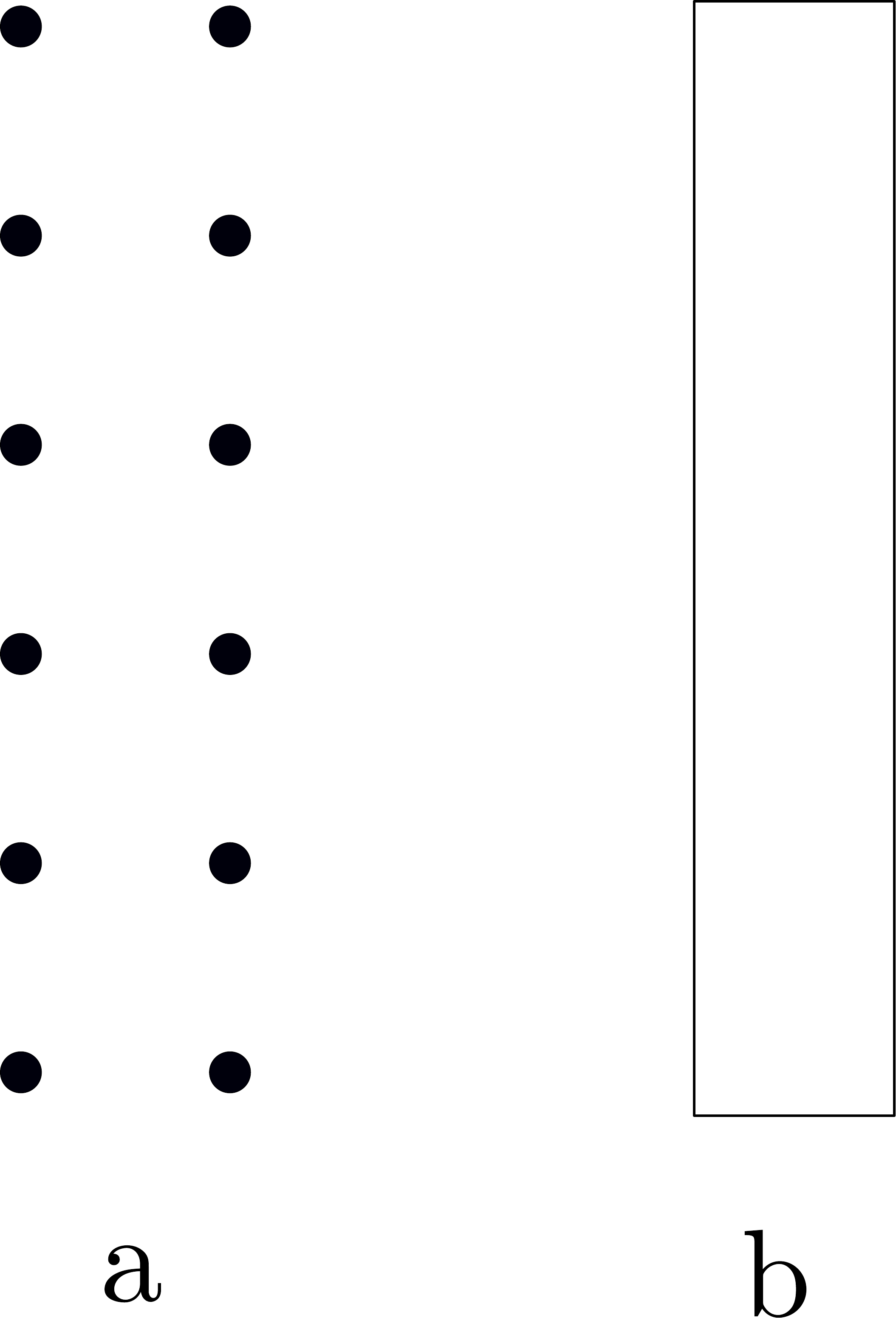}}
\caption{vertical-chain}
\label{fig:vertical-chain}
\end{center}
\end{figure}

\begin{definition}[The \ecp]
In this problem, we are given a set $U$ and a collection $S$ of subsets of $U$. The objective is to find a subset of elements of $S$ that are pairwise disjoint and that cover all the elements of $U$. This problem is one of the 21 problems shown by Karp to be NP-complete \cite{karp1972reducibility}. In fact the reduction from 3-SAT to the \ecp{} given by Karp also implies that this problem is ETH-hard.
\end{definition}

\subsubsection{The construction}
First, we define a \emph{cell} to be any square in $\mathbb{R}^2$ of side $2a+11$, where $a = 20$. Next we define the following point configurations within a cell: let an $X$-cell consist of a $H$ gadget, an $A$ gadget and $2$-chains, let a $Y$-cell consist of a $2$-chain and let $Z$-cells and $Z'$-cells consist of two parallel vertical-chains as depicted in figure \ref{fig:fhg_tsp}.

We are now ready to describe the construction. We start with the pointset described in \ref{subsec:twfractal} and then replace each point with carefully chosen combinations of the above defined gadgets. Since we are in the plane, we have $d=2$. We pick two odd integers $l$ and $v$ such that $l > v$ and $\delta \geq \frac{\log(l(l-v))}{\log(l)} > \delta'$. We can always find such a pair of odd integers. Let $\delta'' = \frac{\log(l(l-v))}{\log(l)}$. Now for all integers $i \geq 0$, consider the pointset $f^{l,v,2}(i)$. We can modify the pointset as follows: we scale all the points about the origin by a factor of $2a + 11$. Then we replace each $f^{l,v,2}(0)$ point by an $X$-cell, each $h^{l,v,2}(0)$ point by a $Y$-cell, and each $g^{l,v,2}(0)$ point by a $Z$-cell. We denote the resulting pointset by $F^{l,v,2}(i)$.

Given an instance of the \ecp{} with a universe $U =\{ u_1,\ldots , u_m \}$ and a collection $S = \{T_1, \ldots ,T_m \}$ of subsets of $U$ such that $|S| = |U| = m$, we reduce it into an instance of TSP on a set of points $P \subset \mathbb{R}^2$ with $|P| = n = O \left(m^{\frac{\delta''}{\delta'' - 1}} \right)$ and $\dimF(P) = \delta''$. The required pointset is obtained by taking the pointset $F^{l,v,2} \left( \lc \frac{\log(m)}{\log(l-v)} \rc \right)$, removing the first row of $H$ gadgets and introducing $1$-chains connecting the rows of $2$-chains as depicted in Figure \ref{fig:f_2_tsp}. From the construction, $F^{l,v,2}( \lc \frac{\log(m)}{\log(l-v)} \rc)$ consists of $m$ rows of $2$-chains $R_1, R_2, \ldots , R_m$. These correspond to the sets in $S$ and the columns correspond to the elements of $U$. 

Finally, just as in the construction by Papadimitriou, we replace the first $A$ gadget above each $H$ gadget with a copy of gadget $B$ if the element of $U$ corresponding to that column is in the set of $S$ corresponding to the row of the chosen $H$ gadget. Furthermore we also replace the copies of $Z$-cells between the $H$ gadget and the newly added $B$ gadget with copies of $Z'$-cells. This gives us our final pointset $P$. We will denote by $H_{i,j}$ the $j$\textsuperscript{th} $H$ gadget from the left in the $i$\textsuperscript{th} row of $H$ gadgets.

\begin{lemma}\label{lem:fractaldim3}
$\dimF(P) \leq \delta''$.
\end{lemma}
\begin{proof}
We observe that $P$ is obtained by merely scaling the pointset $f^{l,v,2}(i)$ and replacing each point by at most $k$ points all within a distance of $c$ of the original point, where $k$ and $c$ are constants. Since $f^{l,v,2}(i)$ has fractal dimension $\delta''$ for our choice of $l$ and $v$, applying Lemmas \ref{lem:scaling} and \ref{lem:substitution} gives us the required bound on the fractal dimension of $P$.
\end{proof}

Before we prove the main result we restate the following lemma 
from \cite{papadimitriou1977euclidean} which we will reuse in our analysis.

\begin{lemma}[\cite{papadimitriou1977euclidean}]\label{lem:papadimitrioutechnical}
Suppose that in an instance $E$ of the TSP we have $N$ $a$-components $G_1,G_2,\ldots,G_N$, such that the distance between any two components is at least $2a$, and $G_0$, the remaining part of $E$, is $a$-compact. Suppose that any optimal Traveling Salesman path of $E$ has its endpoints on $G_0$, and that they do not contain links between any two $a$-components of $E$. Let $L_1, \ldots,L_N$ be the lengths of the optimal $1$-paths of $G_1, \ldots, G_N$ and $L_0$ the length of the optimal $(N+1)$-path of $G_0$. If there is a $1$-path $P$ of $E$ consisting of the union of an optimal $(N+1)$-path of $G_0$, $N$ optimal $1$-paths of $G_1, \ldots, G_N$ and $2N$ edges of length $a$ connecting $a$-components to $G_0$, then $P$ is optimal. If no such $1$-path exists, the optimal $1$-path of $E$ has length greater than $L = L_0 + L_1 + \ldots + L_N + 2Na$.   
\end{lemma}

\begin{proof}[Proof of Theorem \ref{thm:tsplowerbound}]
We reduce the \ecp{} to TSP. The pointset required is $P$. Just as in the proof of Papadimitriou \cite{papadimitriou1977euclidean}(Theorem 2 ), we have that whenever a TSP path traverses an $A$ gadget it can also optimally visit a $H$ gadget above or below $A$. However when the path traverses a $B$ gadget it optimally visits a $H$ gadget above or below $B$ if and only if the corresponding $2$-chain is traversed in mode 1. The $m(m-1)$ $H$ gadgets are $a$-components. The remaining points in $P$ are $a$-compact as there is a path traversing them such that every edge has length at most $8$. This implies that we may apply Papadimitriou's technical Lemma \ref{lem:papadimitrioutechnical} from \cite{papadimitriou1977euclidean}. We have that $N = m(m-1)$, and $L_1 = \ldots = L_m = 32$. The shortest $1$-path of the $a$-compact portion of $P$ obtained after removing the $H$-gadgets has length $L_0 = m^2 (3a + 21) + m(4a + 13 + \sqrt{2}) + 2m -2a + 11$. Suppose an optimal path $Q$ as described in Lemma \ref{lem:papadimitrioutechnical} exists. Then $Q$ must visit each $2$-chain of $P$ in either mode $1$ or mode $2$. Furthermore $Q$ must visit a configuration $H$ from every configuration $A$ to maintain optimality. If $Q$ visits a configuration $B$ in mode $2$ then it can not visit a vertically adjacent configuration $H$. Since there are $m-1$ rows of configuration $H$ and $n$ rows of configurations $A$ and $B$ it must be that $Q$ traverses exactly one configuration $B$ in mode $2$ for every column. This implies that the sets corresponding to the rows of $2$-chains that are traversed in mode $2$ cover $U$ exactly. In the other direction if a solution to the instance of \ecp{} exists then we may traverse the rows corresponding to the sets in the solution in mode $2$ to obtain an optimal path as described in Lemma 2 of \cite{papadimitriou1977euclidean}. This implies that if the Euclidean TSP path instance can be solved in time $2^{O \left(n^{ \left(1- \frac{1}{\delta'} \right)} \right)} = 2^{o \left(n^{ \left(1- \frac{1}{\delta''} \right)} \right)}$ then the corresponding instance of the \ecp{} can be solved in time $2^{o(m)}$ and ETH fails.
\end{proof}

\begin{figure}[H]
\begin{center}
\scalebox{0.12}{\includegraphics{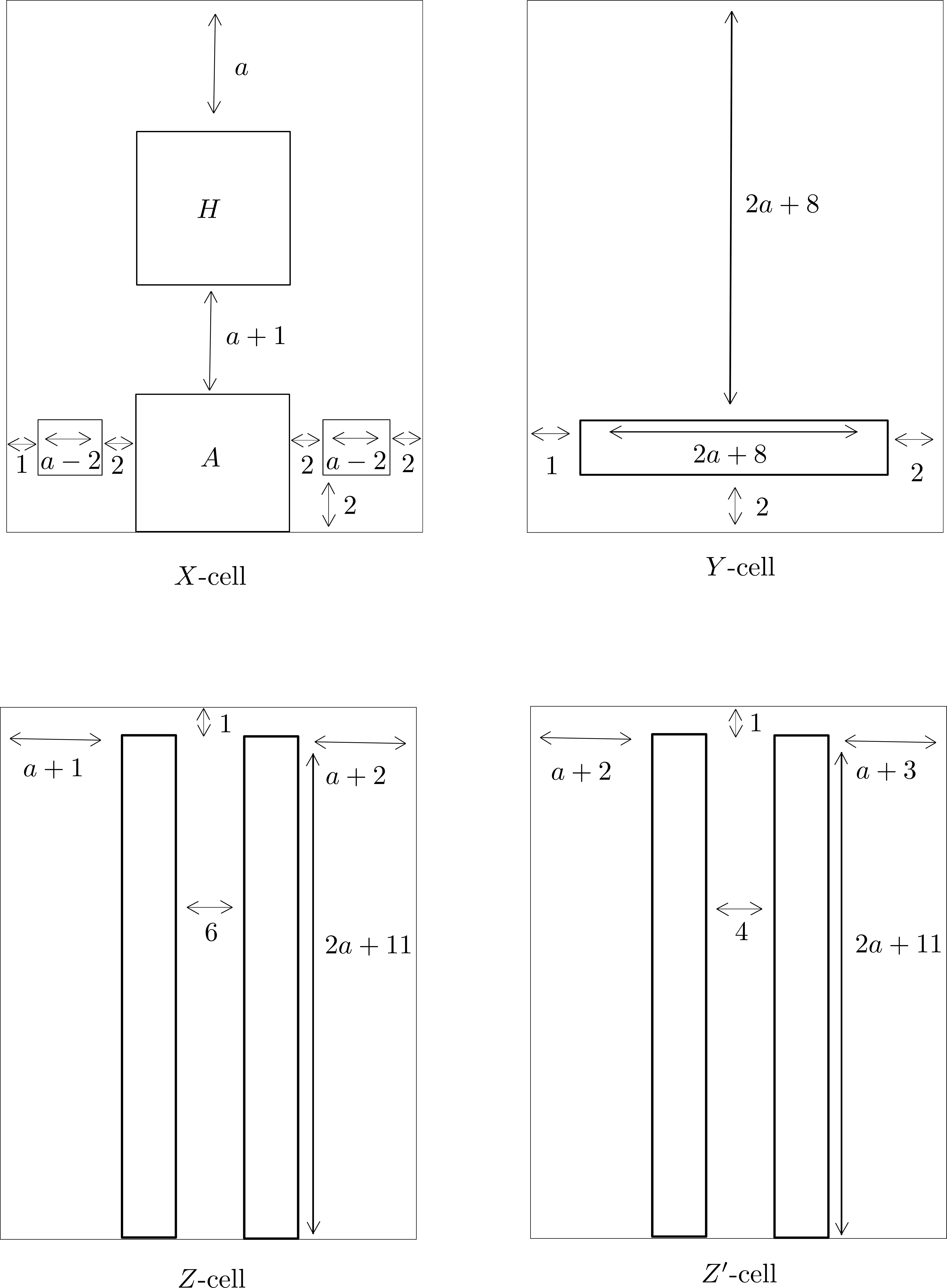}}
\caption{The different types of cells used in the construction}
\label{fig:fhg_tsp}
\end{center}
\end{figure}

\begin{figure}[H]
\begin{center}
\scalebox{0.10}{\includegraphics{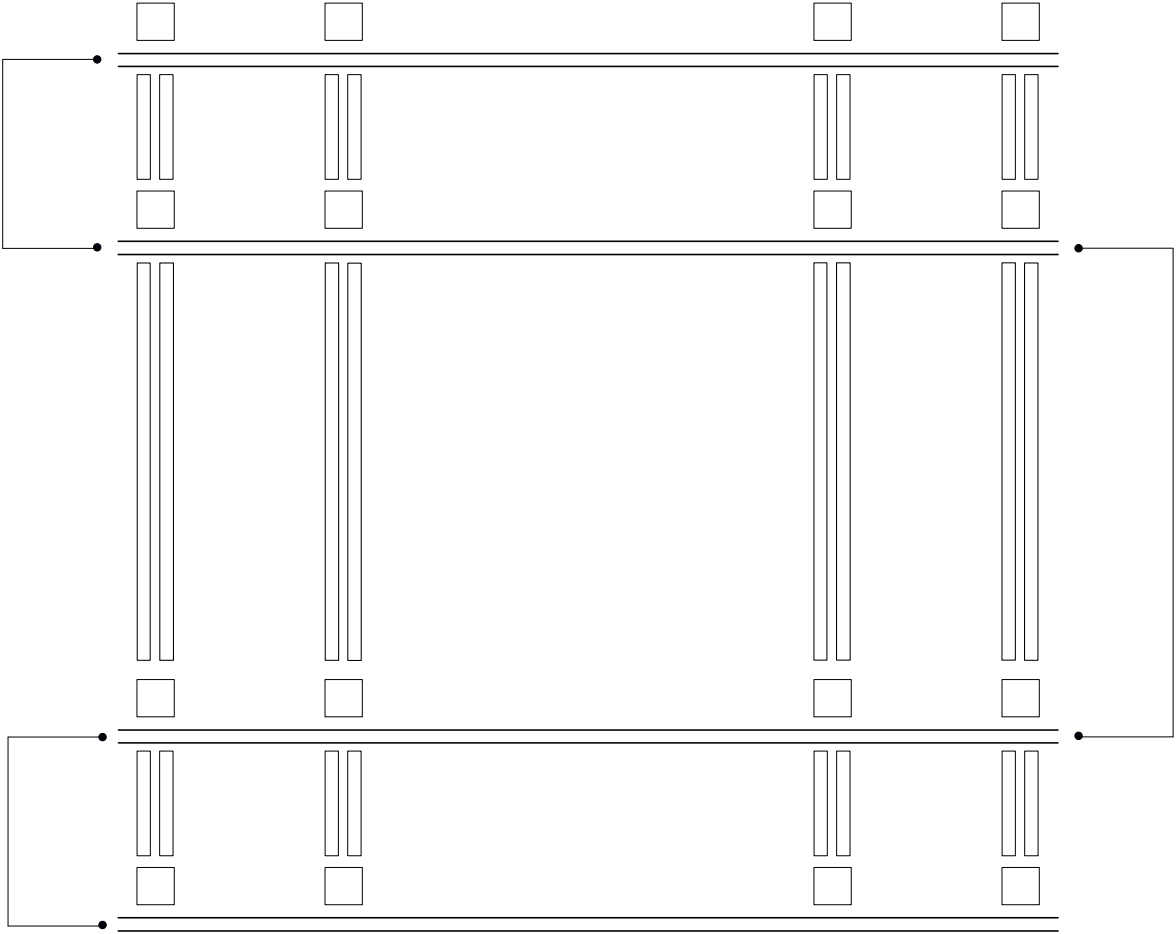}}
\caption{The pointset obtained from $F^{3,1,2}(2)$}
\label{fig:f_2_tsp}
\end{center}
\end{figure}

\subsection{Lower bound for TSP in $\mathbb{R}^d$}\label{subsec:lowerboundtsparbitdim}
In this Subsection we show a running time lower bound for TSP, on pointsets of arbitrary fractal dimension in $\mathbf{R}^d$, for $d \geq 3$.

\begin{theorem}\label{thm:tsplowerboundarbitrarydim}
Let $d \geq 3$ be some integer. Then, for all $\delta \in (2,d)$, for all $\delta'<\delta$ and for all $n_0 \in \mathbb{N}$, if there exists $n \geq n_{0}$ such that Euclidean TSP in $\mathbb{R}^d$ on all pointsets of size $n$ and fractal dimension at most $\delta$ can be solved in time $2^{O(n^{1-1/\delta'})}$, then ETH fails.
\end{theorem}

We prove the above theorem using a similar argument to that of Marx and Sidiropoulos \cite{MS14} where they obtain a lower bound for the running time of Euclidean TSP in $\mathbb{R}^d$. The proof uses a reduction of an instance of the \csp{} to a TSP instance. Again we use the earlier mentioned gadgets introduced by Papadimitriou \cite{papadimitriou1977euclidean} in our construction. We will also use the following theorem from \cite{MS14}.

\begin{theorem}\label{lem:cspeth}
Let $I$ be an instance of the \csp{}. For every fixed $d \geq 2$, there is no $f(|V|)|I|^{o(|V|^{1-1/d})}$ algorithm for CSP$({\cal{R}}_{d})$ for any function $f$, unless ETH fails.
\end{theorem}

\subsubsection{Reduction of CSP to TSP}
We now show that for all $\delta \in (2,d)$ and for all $\delta' < \delta$, it is possible to pick a corresponding $ \delta' < \delta'' < \delta$ and obtain a reduction of an instance of the \csp{} to an instance of TSP on a pointset with fractal dimension $\delta''$. Let $I = (V,D,C)$ be an instance of the \csp{} with  primal graph $R[N,d]$. We arbitrarily pick a pair of positive integers $l$ and $v$ such that $\delta' < \frac{\log l(l-v)^{d-1}}{\log l} < \delta$. Let $ \delta'' = \frac{\log l(l-v)^{d-1}}{\log l}$. Let $k = \log_{l-v}  N $. We can always assume that $k$ is an integer because otherwise we can just add more variables and constraints that will not affect the solution of $I$, until $N$ is large enough. This will only change the size of the instance by at most a factor of $(l-v)$ which is a constant.

\paragraph{The base pointset}.~Having fixed our choices for $l$, $v$ and $k$ we now refer to the construction used to prove theorem \ref{thm:spannerhigherdim}. First we define two classes of points. We call the point of $f^{l,v,d}(0)$ a \emph{grid point}. Also for all $i \in \{1, \ldots , d \}$ we call the point of $h^{l,v,d}_i(0)$ a \emph{connection point}. Since for all $i$, the pointset $f^{l,v,d}(i)$ is constructed inductively we have that every point of $f^{l,v,d}(i)$ is either a grid point or a connection point. Let $P$ be a rescaling of the points of $f^{l,v,d}(k)$ by a factor of $\gamma |D|$ where $\gamma$ will be determined later. Let $P_g \subset P$ be the set of grid points of $P$ and $P_c \subset P$ be the set of connecting points of $P$. For all $i \in \{1, \ldots , d \}$, let $P_i \in P_c$ be the set of $h^{l,v,d}_i(0)$ points in $P_c$. For all $x \in P$, we denote by $\dbox(x)$, the $d$-dimensional hypercube of side length $\gamma |D|$ centered at $x$. We will also define a spanner $G$ with vertex set $V(G) = P$. For all $i \in \{1, \ldots , d \}$ and all $u,v \in P_i$ such that the coordinates of $u$ and $v$ differ only along the $i$ \textsuperscript{th} dimension by exactly $\gamma |D|$ we have that the edge $(u,v) \in E(G)$. For all $u \in P_g$ and $v \in P$ such that the coordinates of $u$ and $v$ differ only along a single dimension by exactly $\gamma |D|$ we have that the edge $(u,v) \in E(G)$. Similar to the proof of Theorem \ref{thm:spannerhigherdim}, we observe that contracting all edges of $G$ incident to vertices in $P_c$ gives a graph $G'$ which is an $(l-v)^{kd}$ or $N^d$ $d$-dimensional grid.

\paragraph{Substituting gadgets}.~Next for all $u \in P$, we consider $\dbox(u)$ and inside this $d$-dimensional hypercube we replace $u$ with points corresponding to certain gadgets to get a TSP instance $(X, \alpha)$ similar to the approach used in subsection \ref{subsec:lowerboundtsp}. Intuitively these gadgets will be used to represent the variables and constraints of the \csp{} instance $I$. We will also reuse the gadgets configurations A, B and H described in the proof of theorem \ref{thm:tsplowerbound}. We now describe the following generalized version of $1$-chains for higher dimensions and a variant of $2$-chains called flexible $2$-chains that we will need for our construction. In what follows we use $\lVert \cdot \rVert_{1}$ to denote the $\ell_{1}$ distance.

\begin{definition}[$1$-chain]
It consists of a sequence of points $\{ x_1,x_2, \ldots, x_k \}$ in $\mathbb{R}^d$ where for all $i \in [k-1]$ $x_i$ and $x_{i+1}$ differ in exactly one coordinate and $\lVert x_i - x_{i+1} \rVert_{1} = 1$. Moreover for all $i,j \in [k]$ such that $|i-j| \leq 20$, we have $\lVert x_i - x_j \rVert_{1} \geq |i-j|$.
\end{definition}

\begin{definition}[Flexible $2$-chain]
It is obtained by starting with a set of point in $\mathbb{R}^d$ that form a $2$-chain $C$ 
with parallel rows $\{ x_1, \ldots, x_k \} $ and $ \{ y_1, \ldots, y_k \}$, 
and then applying any number of the following modification operations to it, while maintaining the invariant that for all $i,j \in [k]$ such that $|i-j| \leq 20$, we have $\lVert x_i - x_j \rVert_{1} \geq \lVert i -  j \rVert$ and $\lVert y_i - y_j \rVert_{1} \geq \lVert i -  j \rVert$.
\begin{enumerate}
\item Pick some $i \in [k]$ and rotate $\{ x_i, \ldots, x_k \}$
and $\{ y_i, \ldots, y_k \}$ by $\frac{\pi}{2}$ along the axis through $x_i$ and $y_i$.
\item Pick some $i \in [k]$ and rotate $\{ x_i, \ldots, x_k \}$
and $\{ y_i, \ldots, y_k \}$ by $\frac{\pi}{100}$ along the axis through the midpoints of the segments $x_iy_i$ and $x_{i-1}y_{i-1}$.
\end{enumerate}
\end{definition}

\begin{definition}[Linking $2$-chains]
We say that two flexible $2$-chains \emph{link} if their union results in a longer flexible $2$-chain.
\end{definition} 

Now we add gadgets in a similar manner to the construction used in section $4$ of \cite{MS14}. First we represent the variables. For each variable in $V$, there exists a corresponding point in the primal graph $R[N,d]$. Combined with the earlier observation that $G$ contains an $N^d$ grid this implies a natural mapping from $V$ to $P_g$. We identify the variables $V$ with the corresponding points in $P_g$. Henceforth we use the terms variable, point and vertex interchangeably to refer to the elements of $V$, $P$ and $V(G)$. Each variable $u \in V$ has at most $2d$ variables that it shares constraints with given by its neighbors in $G'$. We denote this set of variables by $N(u)$. Similarly we denote by $S(u)$ the set of vertices in $G$ that share an edge with $u$. Note that $N(u)$ need not be equal to $S(u)$. For all $i \in \{1, \ldots , d \} $, and for all $v \in S(u)$ such that the line passing through $u$ and $v$ is parallel to dimension $i$ we define $\ifac(u,v) = \ifac(v,u)$ to be the line segment of length $\gamma |D|$  parallel to dimension $(i \mod d) + 1$ passing through the center of the $(d-1)$-dimensional hypercube $\dbox(u) \cap \dbox(v)$.

For all $u \in V$ and for all $j \in \{1, \ldots, |D| \}$, we will add the following gadgets to $\dbox(u)$. We add two 2-chains $\Gamma(u,j,1)$ and $\Gamma(u,j,2)$. 
For each neighbor $v \in N(u)$, we add three flexible 2-chains $\Gamma(u,j,v,1)$, $\Gamma(u,j,v,2)$ and $\Gamma(v,j,u)$. We also ensure that for all $v \in N(u)$ and for all $j \in \{1, \ldots , |D| \}$ a pair of end points of the flexible 2-chains $\Gamma(u,j,v,1)$ and $\Gamma(u,j,v,2)$, and both pairs of end points of $\Gamma(v,j,u)$ lie on $\ifac(u,r)$ evenly spaced apart, where $r \in S(u)$ is adjacent to $u$ and lies on the shortest path from $v$ to $u$ in $G$. This is depicted in figure \ref{fig:interface2} for $d=3$. We then fix an arbitrary ordering $\pi : \{ 1, \ldots , 2d \} \to N(u)$ and for all $i <2d$ link the non-interface ends of the flexible 2-chains $\Gamma(u,j,\pi(i),2)$ and $\Gamma(u,j,\pi(i+1),1)$. Next, we link one end of $\Gamma(u,j,1)$ to $\Gamma(u,j,\pi(1),1)$ and link one end of $\Gamma(u,j,2)$ to $\Gamma(u,j,\pi(2d),2)$. We also add $|D| - 1$ copies of configuration $H$ and ensure that there exists a $2$-plane $h(u)$ such that for all $j \in \{1, \ldots , |D| \}$, subsets of $\Gamma(u,j,1)$ are arranged with appended copies of configuration $B$ and alongside copies of configuration $H$ as depicted in figure \ref{fig:vars_temp}. This will ensure that in the optimal solution, $\Gamma(u,j,1)$ is traversed in mode $2$ for exactly one value of $j$ corresponding to a choice of assignment for the variable $u$. 

\begin{figure}[H]
\begin{center}
\scalebox{0.25}{\includegraphics{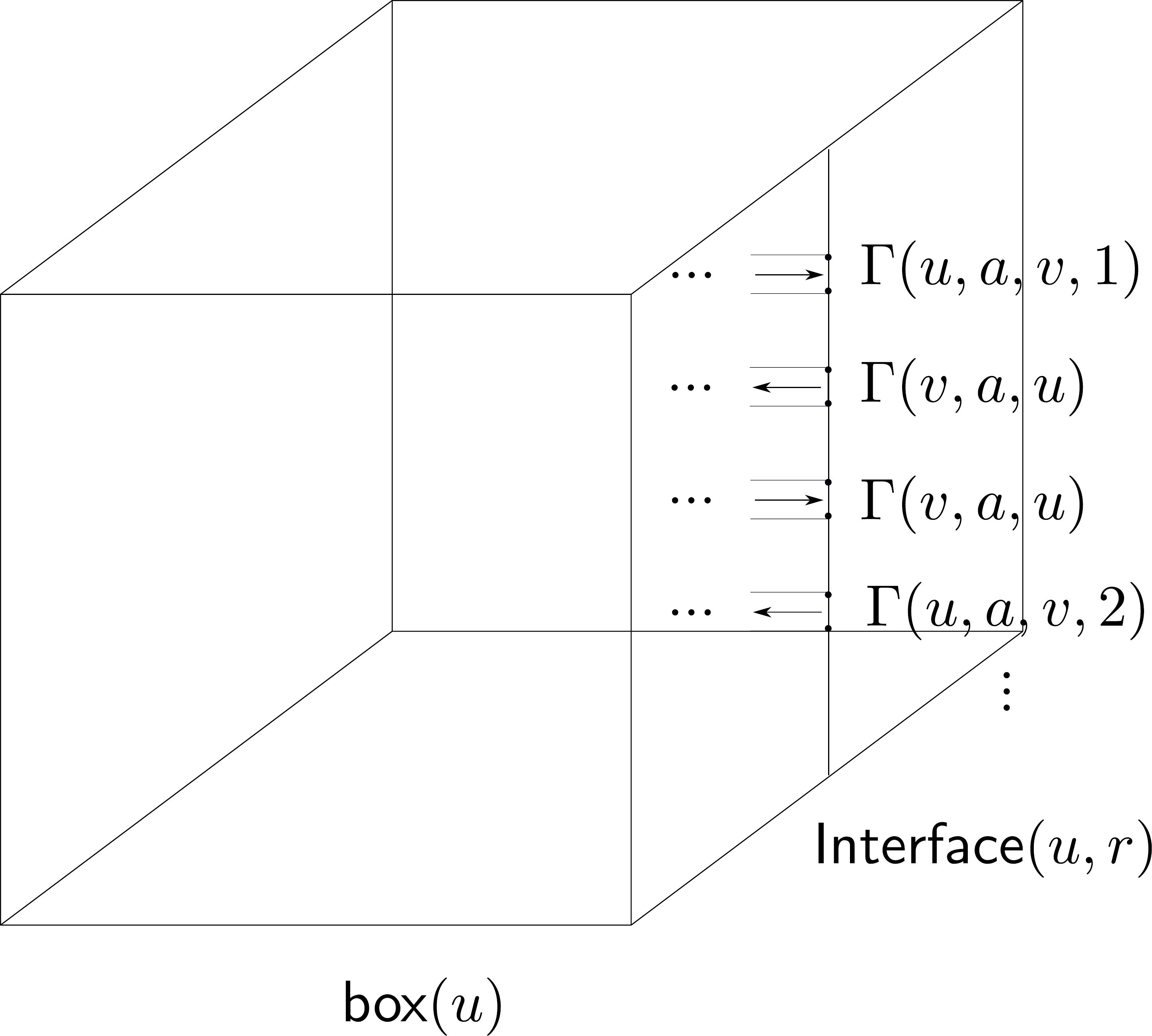}}
\caption{Interface in $\dbox(u)$ for $u,v \in V$ when $d = 3$ }
\label{fig:interface2}
\end{center}
\end{figure}

\begin{figure}[H]
\begin{center}
\scalebox{0.25}{\includegraphics{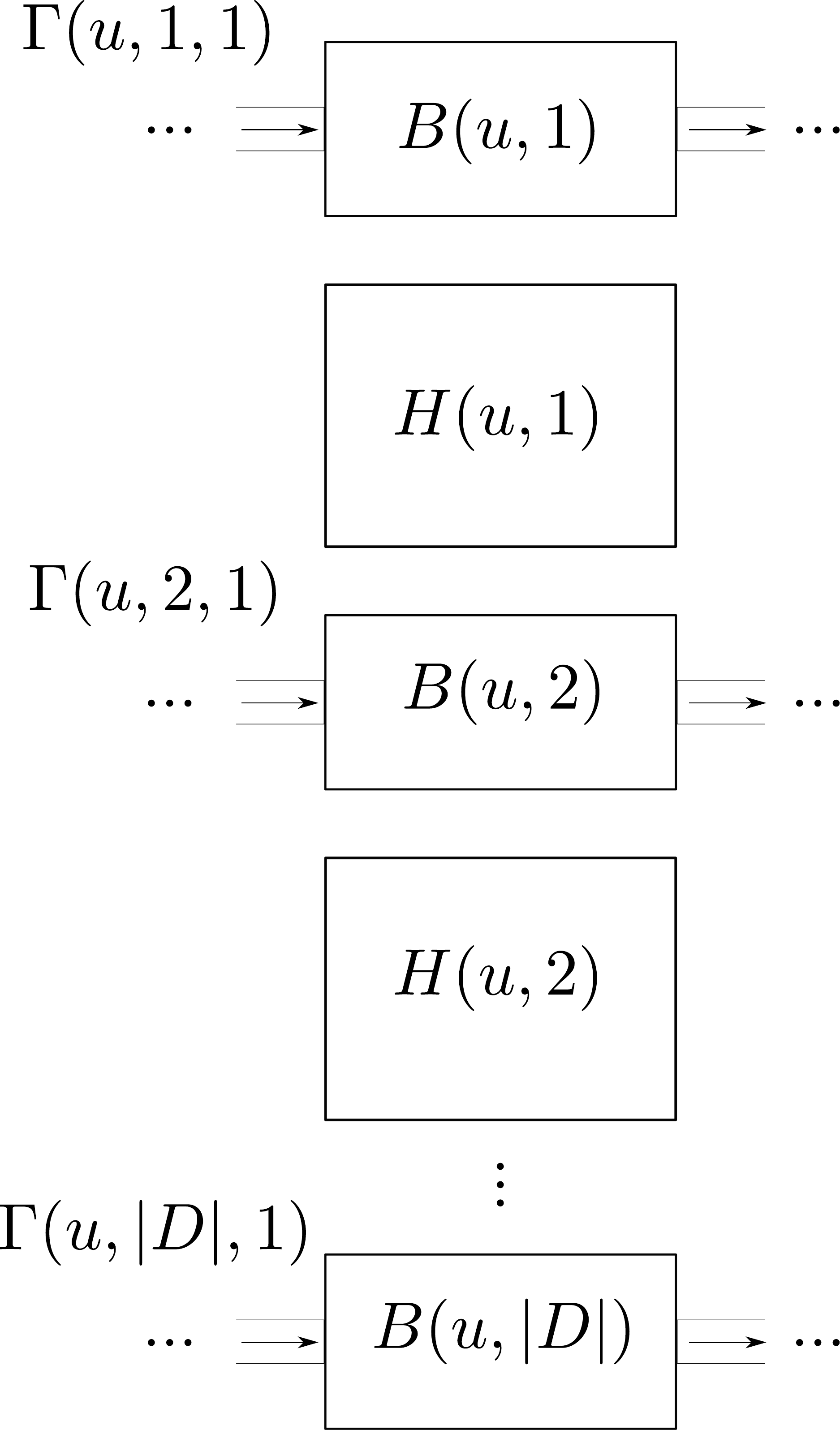}}
\caption{Gadgets corresponding to the variables of the CSP}
\label{fig:vars_temp}
\end{center}
\end{figure}

Next for each constraint $\langle (u,v) , R \rangle$ and each $(a,b) \in D^2$ such that $(a,b) \not \in R$, we add to $\dbox(u)$, two new vertices $x(u,a,v,b)$ and $y(u,a,v,b)$ and a new flexible $2$-chain $\Gamma(u,a,v,b)$ from $x(u,a,v,b)$ to $y(u,a,v,b)$. We next introduce three $B$ configurations $B(u,a,v,b,1)$, $B(u,a,v,b,2)$ and $B(u,a,v,b,3)$. Finally we add two copies of configuration-$H$ $H(u,a,v,b,1)$ and $H(u,a,v,b,2)$. We arrange these newly added gadgets along with subsets of the flexible 2-chains $\Gamma(u,a,v,1)$ and $\Gamma(v,b,u)$ such that they are in a plane $h(u,v,a,b)$ as in figure \ref{fig:constraints_temp}. This ensures that in any optimal solution at most one of $\Gamma(u,a,1)$ and $\Gamma(v,b,1)$ may be traversed in mode 2. Note that the shared constraints are re-encoded into similar gadgets in $\dbox(v)$ swapping the parameters $u$ and $a$ for $v$ and $b$ respectively when naming them. This is done merely for ease of presentation even though it is redundant in terms of the actual reduction argument.

\begin{figure}[H]
\begin{center}
\scalebox{0.25}{\includegraphics{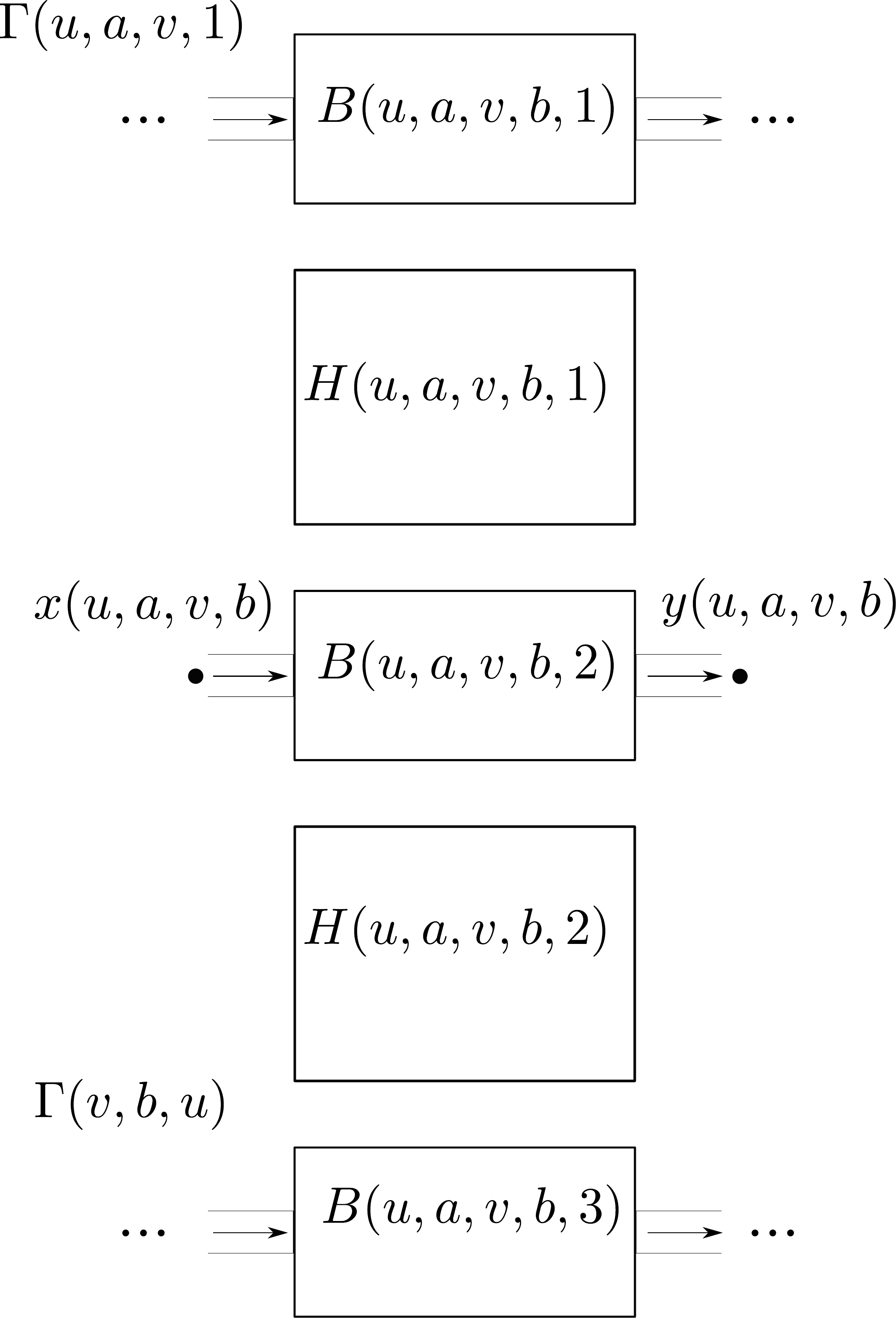}}
\caption{Gadgets corresponding to the constraints of the CSP}
\label{fig:constraints_temp}
\end{center}
\end{figure}

Furthermore, for all $i \in \{1, \ldots , d \} $ and for every point $u \in P_i$ with $S(u) = \{l,r\}$, we will populate $\dbox(u)$ with $4 |D|$ parallel co-planar $2$-chains that run from $\ifac(u,l)$ to $\ifac(u,r)$. For all adjacent $u,v \in V$ that differ in coordinate $i$ and for all $a \in \{1, \ldots, |D| \}$ these $2$-chains will ensure that the ends of $\Gamma(u,a,v)$ in $\dbox(v)$ connect to $\Gamma(u,a,v,1)$ and $\Gamma(u,a,v,2)$ in $\dbox(u)$. When $l,r \in P_i$, these $2$-chains will link with the corresponding 2-chains in $\dbox(l)$ and $\dbox(r)$ at $\ifac(u,l)$ and $\ifac(u,r)$ respectively. This is depicted in figure \ref{fig:interface_temp} when $d = 3$.

\begin{figure}[H]
\begin{center}
\scalebox{0.25}{\includegraphics{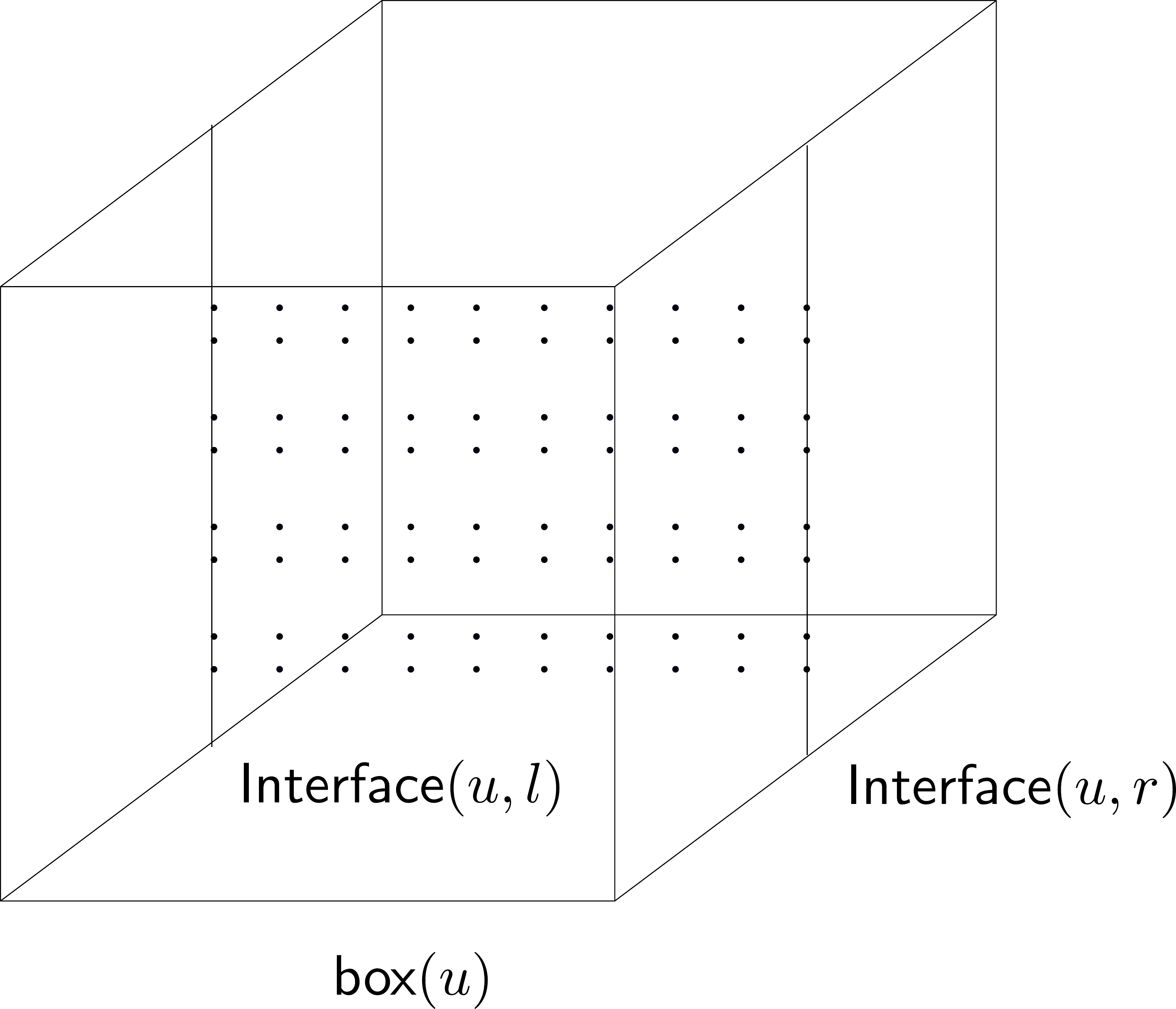}}
\caption{Interface in $\dbox(u)$ for $u,l,r \in P_i$ when $d = 3$ }
\label{fig:interface_temp}
\end{center}
\end{figure}

Finally, we add $1$-chains that connect the flexible $2$-chains so that the optimal solution traverses all $2$-chains. To this end, we consider a Hamiltonian path $h$ in $G$ and consider the vertices of $G$ in the order visited by $h$. When considering a vertex $u$, we add 1-chains appending all the flexible 2-chains $\Gamma(u,j,1)$ and $\Gamma(u,j,2)$ and then we add 1-chains to append all $\Gamma(u,a,v,b)$ that have not already been appended by 1-chains. Consider two arbitrary corner variables $u,w \in V$, we introduce two points $p$ and $q$, and add a $1$-chain from $p$ to $\Gamma(u,1,1)$, and a 1-chain from $\Gamma(w,|D|,2)$ to $q$ such that the following is true.

\begin{lemma}\label{lem:tspendpts}
Any optimal Traveling Salesperson path in the constructed instance has endpoints $p$ and $q$.
\end{lemma}

\subsubsection{Analysis}

\begin{lemma}\label{lem:reductioncsptotsp}
Let $d \geq 3$. Let $I = (V,D,C)$ be an instance of the \csp{} with constraint graph $G = R[N,d]$. Then it is possible to compute an instance $(X,\alpha)$ of TSP with $\dimF(X) = \delta'$ and $|X| \leq N^{(d-1) \frac{\delta}{\delta - 1}} |D|^{O(1)}$, such that the length of the shortest TSP tour of $X$ is at most $\alpha$ if and only if $I$ is satisfiable.
\end{lemma}

\begin{proof}
Our proof follows an identical approach to that of Marx and Sidiropoulos \cite{MS14} (Lemma 4.5). Given an instance of the \csp{} $I$, we construct a TSP instance $(X,\alpha)$ as described in the construction above. Let $G_1 , \ldots G_m$ denote the $H$ configurations in $P'$ and let $G_0 = P' \setminus \cup_{i = 1}^m G_i$. Let a = 20. Leaving enough space between the different gadgets in $X$ will ensure that the following properties are maintained:

\begin{description}
\item{1.} Every configuration $H$ in $X$ is an $a$-component.
\item{2.} Any two copies of configuration $H$ in $X$ are at least distance $2a$ apart.
\item{3.} Every optimal TSP path of $X$ has its endpoints in $G_0$; this follows from Lemma \ref{lem:tspendpts}.

\item{4.} $G_0$ is $(a, m+1)$-compact. We can ensure that this is satisfied by enforcing the following conditions: 1) Distinct $2$-chains and flexible $2$-chains that are not linked are at least distance $20$ from each other. 2) Distinct $1$-chains are at least distance $20$ from each other. 3) $1$-chains and flexible $2$-chains that don't share an end point are at least distance $20$ from each other. 4) Let $C$ be a $1$-chain that shares an endpoint $p$
with some flexible $2$-chain $C_0$. Then, there exists a line $l$ in $\mathbb{R}^d$, with $p \in  l$, such that the $20$ segments in $Q$ closest to $p$ lie in $l$, and the $20$ points in $C$ closest to $p$ also lie in $l$; all other points in $C$ are at distance at least $20$ from $C_0$.
\end{description}

The above conditions can be achieved by setting $\gamma$ to be large enough since the flexible $2$-chains can turn and orient to any axis. Suppose now that $I$ is satisfiable, and let $F : V \to \{1, \ldots , |D| \}$ be a satisfying assignment for $I$. We can build a Traveling Salesperson path $H$ for $X$ as follows. We start at $p$, and we traverse the $1$-chain originating at p. When we reach a 2-chain $C$, we distinguish between the following cases:

1) Suppose that $C$ is a $2$-chain of the form $\Gamma(u,i,1)$ for some $u \in V$, and $i \in [|D|]$. If $F(u) = i$, then we begin traversing $C$ in mode $2$, and if $F(u) \neq i$, then we begin traversing $C$ in mode $1$. 

2) Suppose next that $C$ is a $2$-chain of the form $\Gamma(u, i, v, j)$ for some $u, v \in V$, and $i, j \in \{1, \ldots , |D| \}$. This implies that there exists a constraint $\langle (u, v),R \rangle$, with $(i, j) \not \in R$. Since $F$ satisfies $I$, it follows that either $F(u) \neq i$, or $F(v) \neq j$. If $F(u) \neq i$, and $F(v) \neq j$, then we traverse $\Gamma(u, i, v, j)$ in mode $2$. If $F(u) = i$, then we traverse $\Gamma(u, i, v, j)$ in mode $1$, and we traverse $B(u, i, v, j, 2)$ together with $H(u, i, v, j, 1)$. If $F(v) = j$, then we traverse $\Gamma(u, i, v, j)$ in mode $1$, and we traverse $B(u, i, v, j, 2)$ together
with $H(u, i, v, j, 2)$.

Upon reaching the unlinked end of a flexible $2$-chain, we traverse the $1$-chain that is attached to it, and we
proceed to the next unlinked end of a flexible $2$-chain. Eventually, we reach $q$, and the path terminates. It is immediate
to check that the resulting path $Q$ visits all points in $X$. This gives us the following property:

\begin{description}
\item{5.} The path $H$ consists of the union of an optimal $(m + 1)$-path for $G_0$, $m$ optimal $1$-paths for $G_1, \ldots,G_m$, and $2m$ edges of length $a$ connecting $a$-components to $G_0$.
\end{description}

We now show the converse, that given an optimal path we can obtain a satisfying assignment to the instance of the \csp{}. Any optimal path $Q$ must consist of an optimal $(m + 1)$-path for $G_0$, $m$ optimal $1$-paths for $G_1,\ldots,G_m$, and $2m$ edges of length $a$ connecting $G_0$ with $G_1,\ldots,G_m$ (two edges for each $a$-component). By Lemma 1 of \cite{papadimitriou1977euclidean}, we have that for each $i \in \{ 1,\ldots,m \}$, the configuration-$H$ $G_i$ is visited in such a way that $Q$ must contain two edges between
each configuration-$H$ $G_i$ and some configuration-$B$ that is attached to some flexible $2$-chain $C$, such
that $G_i$ is a neighbor of $C$, and $C$ is traversed in mode $1$. This implies by the construction that
for each $u \in V$, exactly one of the $2$-chains $\Gamma(u, j,1)$ is traversed in mode $2$. We can construct
an assignment $F$ for $I$ by setting $F(u) = j$, where $\Gamma(u, j,1)$ is traversed in mode $2$. It remains
to check that $F$ satisfies $I$. Let $\langle (u, v),R \rangle$ be a constraint, and let $(i, j) \not \in R$. Again, by the
construction, we have that at least one of the flexible $2$-chains $\Gamma(u, i,1), \Gamma(v, j,1)$ is traversed in mode $2$, and therefore either $F(u) \neq i$, or $F(v) \neq j$. It follows that F satisfies $I$.

Finally we obtain an upper bound for $|X|$. Since we start with a rescaled instance of $f^{l,v,d}(k)$ and merely replace each point with at most $|D|^{O(1)}$ points, it follows that $|X| \leq N^{(d-1) \frac{\delta}{\delta - 1}} |D|^{O(1)}$. 
\end{proof}

\begin{lemma}\label{lem:csptotspfractal}
$\dimF(X) \leq \delta''$
\end{lemma}
\begin{proof}
We obtain $X$ by rescaling $f^{l,v,d}(k)$ and then replacing every point $p \in f^{l,v,d}(k)$ by some gadget containing at most $|D|^{O(1)}$ points that are all within the ball of radius $\beta$ of $p$ for some constant $\beta < 50$. Applying Lemmas \ref{lem:scaling} and \ref{lem:substitution}, the statement of the lemma follows.
\end{proof}
\begin{proof}[Proof of Theorem \ref{thm:tsplowerboundarbitrarydim}]
It follows by Theorem \ref{lem:cspeth} and Lemmas \ref{lem:reductioncsptotsp} and \ref{lem:csptotspfractal}.
\end{proof}

\begin{proof}[Proof of Theorem \ref{thm:tsplowerboundmain}]
It follows by Theorems \ref{thm:tsplowerbound} and \ref{thm:tsplowerboundarbitrarydim}.
\end{proof}

\bibliography{bibfile}

\begin{thebibliography}{10}

\bibitem{alber2002geometric}
Jochen Alber and Ji{\v{r}}{\'\i} Fiala.
\newblock Geometric separation and exact solutions for the parameterized
  independent set problem on disk graphs.
\newblock In {\em Foundations of Information Technology in the Era of Network
  and Mobile Computing}, pages 26--37. Springer, 2002.

\bibitem{bartal2012traveling}
Yair Bartal, Lee-Ad Gottlieb, and Robert Krauthgamer.
\newblock The traveling salesman problem: low-dimensionality implies a
  polynomial time approximation scheme.
\newblock In {\em Proceedings of the forty-fourth annual ACM symposium on
  Theory of computing}, pages 663--672. ACM, 2012.

\bibitem{chan2005hierarchical}
Hubert~TH Chan, Anupam Gupta, Bruce~M Maggs, and Shuheng Zhou.
\newblock On hierarchical routing in doubling metrics.
\newblock In {\em Proceedings of the sixteenth annual ACM-SIAM symposium on
  Discrete algorithms}, pages 762--771. Society for Industrial and Applied
  Mathematics, 2005.

\bibitem{chan2009small}
T-H~Hubert Chan and Anupam Gupta.
\newblock Small hop-diameter sparse spanners for doubling metrics.
\newblock {\em Discrete \& Computational Geometry}, 41(1):28--44, 2009.

\bibitem{chekuri2016polynomial}
Chandra Chekuri and Julia Chuzhoy.
\newblock Polynomial bounds for the grid-minor theorem.
\newblock {\em Journal of the ACM (JACM)}, 63(5):40, 2016.

\bibitem{cole2006searching}
Richard Cole and Lee-Ad Gottlieb.
\newblock Searching dynamic point sets in spaces with bounded doubling
  dimension.
\newblock In {\em Proceedings of the thirty-eighth annual ACM symposium on
  Theory of computing}, pages 574--583. ACM, 2006.

\bibitem{Diestel12}
Reinhard Diestel.
\newblock {\em Graph Theory, 4th Edition}, volume 173 of {\em Graduate texts in
  mathematics}.
\newblock Springer, 2012.

\bibitem{gottlieb2008optimal}
Lee-Ad Gottlieb and Liam Roditty.
\newblock An optimal dynamic spanner for doubling metric spaces.
\newblock In {\em European Symposium on Algorithms}, pages 478--489. Springer,
  2008.

\bibitem{gupta2012online}
Anupam Gupta and Kevin Lewi.
\newblock The online metric matching problem for doubling metrics.
\newblock In {\em International Colloquium on Automata, Languages, and
  Programming}, pages 424--435. Springer, 2012.

\bibitem{har2006fast}
Sariel Har-Peled and Manor Mendel.
\newblock Fast construction of nets in low-dimensional metrics and their
  applications.
\newblock {\em SIAM Journal on Computing}, 35(5):1148--1184, 2006.

\bibitem{hubert2012approximating}
T-H Hubert~Chan and Anupam Gupta.
\newblock Approximating tsp on metrics with bounded global growth.
\newblock {\em SIAM Journal on Computing}, 41(3):587--617, 2012.

\bibitem{karger2002finding}
David~R Karger and Matthias Ruhl.
\newblock Finding nearest neighbors in growth-restricted metrics.
\newblock In {\em Proceedings of the thiry-fourth annual ACM symposium on
  Theory of computing}, pages 741--750. ACM, 2002.

\bibitem{karp1972reducibility}
Richard~M Karp.
\newblock Reducibility among combinatorial problems.
\newblock In {\em Complexity of computer computations}, pages 85--103.
  Springer, 1972.

\bibitem{Kozawaetal10}
Kyohei Kozawa, Yota Otachi, and Koichi Yamazaki.
\newblock The carving-width of generalized hypercubes.
\newblock {\em Discrete Math.}, 310(21):2867--2876, November 2010.

\bibitem{krauthgamer2005black}
Robert Krauthgamer and James~R Lee.
\newblock The black-box complexity of nearest-neighbor search.
\newblock {\em Theoretical Computer Science}, 348(2):262--276, 2005.

\bibitem{krauthgamer2006algorithms}
Robert Krauthgamer and James~R Lee.
\newblock Algorithms on negatively curved spaces.
\newblock In {\em 2006 47th Annual IEEE Symposium on Foundations of Computer
  Science (FOCS'06)}, pages 119--132. IEEE, 2006.

\bibitem{krauthgamer2005measured}
Robert Krauthgamer, James~R Lee, Manor Mendel, and Assaf Naor.
\newblock Measured descent: A new embedding method for finite metrics.
\newblock {\em Geometric \& Functional Analysis GAFA}, 15(4):839--858, 2005.

\bibitem{marx2005efficient}
D{\'a}niel Marx.
\newblock Efficient approximation schemes for geometric problems?
\newblock In {\em European Symposium on Algorithms}, pages 448--459. Springer,
  2005.

\bibitem{marx2010can}
D{\'a}niel Marx.
\newblock Can you beat treewidth?
\newblock {\em Theory OF Computing}, 6:85--112, 2010.

\bibitem{MS14}
D\'{a}niel Marx and Anastasios Sidiropoulos.
\newblock The limited blessing of low dimensionality: When 1-1/d is the best
  possible exponent for d-dimensional geometric problems.
\newblock In {\em Proceedings of the Thirtieth Annual Symposium on
  Computational Geometry}, SOCG'14, pages 67:67--67:76, 2014.

\bibitem{papadimitriou1977euclidean}
Christos~H Papadimitriou.
\newblock The euclidean travelling salesman problem is np-complete.
\newblock {\em Theoretical computer science}, 4(3):237--244, 1977.

\bibitem{robertson1994quickly}
Neil Robertson, Paul Seymour, and Robin Thomas.
\newblock Quickly excluding a planar graph.
\newblock {\em Journal of Combinatorial Theory, Series B}, 62(2):323--348,
  1994.

\bibitem{robertson1986graph}
Neil Robertson and Paul~D Seymour.
\newblock Graph minors. v. excluding a planar graph.
\newblock {\em Journal of Combinatorial Theory, Series B}, 41(1):92--114, 1986.

\bibitem{salowe1991construction}
Jeffrey~S Salowe.
\newblock Construction of multidimensional spanner graphs, with applications to
  minimum spanning trees.
\newblock In {\em Proceedings of the seventh annual symposium on Computational
  geometry}, pages 256--261. ACM, 1991.

\bibitem{sidiropoulos2017algorithmic}
Anastasios Sidiropoulos and Vijay Sridhar.
\newblock Algorithmic interpretations of fractal dimension.
\newblock In {\em 33rd International Symposium on Computational Geometry, SoCG
  2017, July 4-7, 2017, Brisbane, Australia}, volume~77 of {\em LIPIcs}, pages
  58:1--58:16. Schloss Dagstuhl - Leibniz-Zentrum fuer Informatik, 2017.

\bibitem{smith1998geometric}
Warren~D Smith and Nicholas~C Wormald.
\newblock Geometric separator theorems and applications.
\newblock In {\em Foundations of Computer Science, 1998. Proceedings. 39th
  Annual Symposium on}, pages 232--243. IEEE, 1998.

\bibitem{takayasu1990fractals}
Hideki Takayasu.
\newblock {\em Fractals in the physical sciences}.
\newblock Manchester University Press, 1990.

\bibitem{talwar2004bypassing}
Kunal Talwar.
\newblock Bypassing the embedding: algorithms for low dimensional metrics.
\newblock In {\em Proceedings of the thirty-sixth annual ACM symposium on
  Theory of computing}, pages 281--290. ACM, 2004.

\bibitem{vaidya1991sparse}
Pravin~M Vaidya.
\newblock A sparse graph almost as good as the complete graph on points ink
  dimensions.
\newblock {\em Discrete \& Computational Geometry}, 6(3):369--381, 1991.

\end{thebibliography}

\end{document}